\documentclass[a4paper,10pt]{article}

\newcommand{\typeof}{1} %

\usepackage{ifthen}
\newcommand{\condinc}[2]{\ifthenelse{\equal{\typeof}{0}}{#1}{#2}}

\usepackage{mathrsfs}
\usepackage{amscd}
\usepackage{pdfsync} 

\usepackage{mathrsfs}
\condinc{}{}
\usepackage{cmll}
\usepackage{euscript}
\usepackage{amssymb,amsmath}

\newdimen\proofrulebreadth \proofrulebreadth=.05em
\newdimen\proofdotseparation \proofdotseparation=1.25ex
\newdimen\proofrulebaseline \proofrulebaseline=2ex
\newcount\proofdotnumber \proofdotnumber=3
\let\then\relax
\def\hfi{\hskip0pt plus.0001fil}
\mathchardef\squigto="3A3B
\newif\ifinsideprooftree\insideprooftreefalse
\newif\ifonleftofproofrule\onleftofproofrulefalse
\newif\ifproofdots\proofdotsfalse
\newif\ifdoubleproof\doubleprooffalse
\let\wereinproofbit\relax
\newdimen\shortenproofleft
\newdimen\shortenproofright
\newdimen\proofbelowshift
\newbox\proofabove
\newbox\proofbelow
\newbox\proofrulename
\def\shiftproofbelow{\let\next\relax\afterassignment\setshiftproofbelow\dimen0 }
\def\shiftproofbelowneg{\def\next{\multiply\dimen0 by-1 }%
\afterassignment\setshiftproofbelow\dimen0 }
\def\setshiftproofbelow{\next\proofbelowshift=\dimen0 }
\def\setproofrulebreadth{\proofrulebreadth}

\def\prooftree{
\ifnum  \lastpenalty=1
\then   \unpenalty
\else   \onleftofproofrulefalse
\fi
\ifonleftofproofrule
\else   \ifinsideprooftree
        \then   \hskip.5em plus1fil
        \fi
\fi
\bgroup
\setbox\proofbelow=\hbox{}\setbox\proofrulename=\hbox{}%
\let\justifies\proofover\let\leadsto\proofoverdots\let\Justifies\proofoverdbl
\let\using\proofusing\let\[\prooftree
\ifinsideprooftree\let\]\endprooftree\fi
\proofdotsfalse\doubleprooffalse
\let\thickness\setproofrulebreadth
\let\shiftright\shiftproofbelow \let\shift\shiftproofbelow
\let\shiftleft\shiftproofbelowneg
\let\ifwasinsideprooftree\ifinsideprooftree
\insideprooftreetrue
\setbox\proofabove=\hbox\bgroup$\displaystyle 
\let\wereinproofbit\prooftree
\shortenproofleft=0pt \shortenproofright=0pt \proofbelowshift=0pt
\onleftofproofruletrue\penalty1
}

\def\eproofbit{
\ifx    \wereinproofbit\prooftree
\then   \ifcase \lastpenalty
        \then   \shortenproofright=0pt  
        \or     \unpenalty\hfil         
        \or     \unpenalty\unskip       
        \else   \shortenproofright=0pt  
        \fi
\fi
\global\dimen0=\shortenproofleft
\global\dimen1=\shortenproofright
\global\dimen2=\proofrulebreadth
\global\dimen3=\proofbelowshift
\global\dimen4=\proofdotseparation
\global\count255=\proofdotnumber
$\egroup  
\shortenproofleft=\dimen0
\shortenproofright=\dimen1
\proofrulebreadth=\dimen2
\proofbelowshift=\dimen3
\proofdotseparation=\dimen4
\proofdotnumber=\count255
}

\def\proofover{
\eproofbit 
\setbox\proofbelow=\hbox\bgroup 
\let\wereinproofbit\proofover
$\displaystyle
}%
\def\proofoverdbl{
\eproofbit 
\doubleprooftrue
\setbox\proofbelow=\hbox\bgroup 
\let\wereinproofbit\proofoverdbl
$\displaystyle
}%
\def\proofoverdots{
\eproofbit 
\proofdotstrue
\setbox\proofbelow=\hbox\bgroup 
\let\wereinproofbit\proofoverdots
$\displaystyle
}%
\def\proofusing{
\eproofbit 
\setbox\proofrulename=\hbox\bgroup 
\let\wereinproofbit\proofusing
\kern0.3em$
}

\def\endprooftree{
\eproofbit 
  \dimen5 =0pt
\dimen0=\wd\proofabove \advance\dimen0-\shortenproofleft
\advance\dimen0-\shortenproofright
\dimen1=.5\dimen0 \advance\dimen1-.5\wd\proofbelow
\dimen4=\dimen1
\advance\dimen1\proofbelowshift \advance\dimen4-\proofbelowshift
\ifdim  \dimen1<0pt
\then   \advance\shortenproofleft\dimen1
        \advance\dimen0-\dimen1
        \dimen1=0pt
        \ifdim  \shortenproofleft<0pt
        \then   \setbox\proofabove=\hbox{%
                        \kern-\shortenproofleft\unhbox\proofabove}%
                \shortenproofleft=0pt
        \fi
\fi
\ifdim  \dimen4<0pt
\then   \advance\shortenproofright\dimen4
        \advance\dimen0-\dimen4
        \dimen4=0pt
\fi
\ifdim  \shortenproofright<\wd\proofrulename
\then   \shortenproofright=\wd\proofrulename
\fi
\dimen2=\shortenproofleft \advance\dimen2 by\dimen1
\dimen3=\shortenproofright\advance\dimen3 by\dimen4
\ifproofdots
\then
        \dimen6=\shortenproofleft \advance\dimen6 .5\dimen0
        \setbox1=\vbox to\proofdotseparation{\vss\hbox{$\cdot$}\vss}%
        \setbox0=\hbox{%
                \advance\dimen6-.5\wd1
                \kern\dimen6
                $\vcenter to\proofdotnumber\proofdotseparation
                        {\leaders\box1\vfill}$%
                \unhbox\proofrulename}%
\else   \dimen6=\fontdimen22\the\textfont2 
        \dimen7=\dimen6
        \advance\dimen6by.5\proofrulebreadth
        \advance\dimen7by-.5\proofrulebreadth
        \setbox0=\hbox{%
                \kern\shortenproofleft
                \ifdoubleproof
                \then   \hbox to\dimen0{%
                        $\mathsurround0pt\mathord=\mkern-6mu%
                        \cleaders\hbox{$\mkern-2mu=\mkern-2mu$}\hfill
                        \mkern-6mu\mathord=$}%
                \else   \vrule height\dimen6 depth-\dimen7 width\dimen0
                \fi
                \unhbox\proofrulename}%
        \ht0=\dimen6 \dp0=-\dimen7
\fi
\let\doll\relax
\ifwasinsideprooftree
\then   \let\VBOX\vbox
\else   \ifmmode\else$\let\doll=$\fi
        \let\VBOX\vcenter
\fi
\VBOX   {\baselineskip\proofrulebaseline \lineskip.2ex
        \expandafter\lineskiplimit\ifproofdots0ex\else-0.6ex\fi
        \hbox   spread\dimen5   {\hfi\unhbox\proofabove\hfi}%
        \hbox{\box0}%
        \hbox   {\kern\dimen2 \box\proofbelow}}\doll%
\global\dimen2=\dimen2
\global\dimen3=\dimen3
\egroup 
\ifonleftofproofrule
\then   \shortenproofleft=\dimen2
\fi
\shortenproofright=\dimen3
\onleftofproofrulefalse
\ifinsideprooftree
\then   \hskip.5em plus 1fil \penalty2
\fi
}

\condinc{}{}

\newenvironment{varitemize}
{
\begin{list}{\condinc{\labelitemii}{\labelitemi}}
{\setlength{\itemsep}{0pt}
 \setlength{\topsep}{0pt}
 \setlength{\parsep}{0pt}
 \setlength{\partopsep}{0pt}
 \setlength{\leftmargin}{15pt}
 \setlength{\rightmargin}{0pt}
 \setlength{\itemindent}{0pt}
 \setlength{\labelsep}{5pt}
 \setlength{\labelwidth}{10pt}
}}
{
 \end{list} 
}

\newcounter{numberone}

\newenvironment{varenumerate}
{
\begin{list}{\arabic{numberone}.}
{
  \usecounter{numberone}
  \setlength{\itemsep}{0pt}
  \setlength{\topsep}{0pt}
  \setlength{\parsep}{0pt}
  \setlength{\partopsep}{0pt}
  \setlength{\leftmargin}{15pt}
  \setlength{\rightmargin}{0pt}
  \setlength{\itemindent}{0pt}
  \setlength{\labelsep}{5pt}
  \setlength{\labelwidth}{15pt}
}}
{
\end{list} 
}

\newcounter{numbertwo}

\condinc{}{}

\newcommand{\NN}{\mathbb{N}}

\newcommand{\RR}{\mathbb{R}}
\newcommand{\RRI}{\mathbb{R}_{[0, 1]}}

\usepackage{cmll}

\usepackage{proof}
\usepackage{amsfonts}%
\usepackage{graphicx}
\usepackage{tikz}
\usepackage{url}
\condinc{}{\usepackage{a4wide}}

\newtheorem{corollary}{Corollary}

\newcommand{\al}[1]{\mathbb{#1}}

\newcommand{\distrone}{\mathcal{D}}
\newcommand{\distrstringone}{\mathcal{D}}
\newcommand{\distrstringtwo}{\mathcal{E}}
\newcommand{\confeval}{\mathcal{CEV}}
\newcommand{\distrset}[1]{\mathbb{P}_{#1}}

\newcommand{\alcom}[4]{((#1\bullet #2)(#3))(#4)}
\newcommand{\alcomn}[4]{((#1\odot #2)(#3))(#4)}
\newcommand{\alcomnf}[3]{(#1\odot #2)(#3)}

\newcommand{\QQ}{\mathbb{Q}}
\newcommand{\f}[3]{{#1}:{#2} \rightarrow {#3}}
\newcommand{\fun}[2]{#1\rightarrow #2}
\newcommand{\natone}{n}
\newcommand{\primrec}[2]{\mathit{rec}(#1,#2)}
\newcommand{\V}[1]{\mathbf{#1}}
\newcommand{\PrR}{\mathscr{PR}}
\newcommand{\PrC}{\mathscr{PC}}
\newcommand{\PrT}{\mathscr{PT}}
\newcommand{\SrR}{\mathscr{SR}}
\newcommand{\PrPC}{\mathscr{PPC}}
\newcommand{\PPrT}{\mathscr{PPR}}

\newcommand{\id}{\mathit{id}}

\newcommand{\transone}{\delta}
\newcommand{\ptmone}{M}

\newcommand{\conf}[4]{\langle #1,#2,#3,#4\rangle}
\newcommand{\stone}{q}
\newcommand{\ssone}{Q}

\newcommand{\charone}{a}
\newcommand{\chartwo}{b}
\newcommand{\casef}[1]{\mathit{case}(#1)}

\newcommand{\prmone}{R}
\newcommand{\CONFR}[1]{\mathcal{CR}_{#1}}
\newcommand{\FSR}[2]{\mathcal{FCR}_{\!#1}^{#2}}
\newcommand{\INITR}[2]{\mathcal{INR}_{\!#1}^{#2}}
\newcommand{\FUNCR}[1]{\mathcal{IO}_{\!#1}}
\newcommand{\alpone}{\Sigma}
\newcommand{\strone}{s}
\newcommand{\strtwo}{t}
\newcommand{\symone}{a}
\newcommand{\bk}[1]{#1_{b}}

\condinc{}{
\newtheorem{definition}{Definition}
\newtheorem{proposition}{Proposition}
\newtheorem{lemma}{Lemma}
\newtheorem{example}{Example}
\newtheorem{theorem}{Theorem}
\newenvironment{proof}{\begin{trivlist}
       \item[\hskip \labelsep {\bfseries Proof.}]}{\hfill $\Box$ \end{trivlist}}}

\condinc{}{\newcommand{\qed}{\hfill$\Box$}}

\newcommand{\BPP}{\textbf{BPP}}
\newcommand{\ZPP}{\textbf{ZPP}}

\begin{document}

\condinc{
\title{Probabilistic Recursion Theory\\ and Implicit Computational Complexity\thanks{This work is partially supported by the
  ANR project 12IS02001 PACE.}}
\author{Ugo Dal Lago \and Sara Zuppiroli}
\institute{Universit\`a di Bologna \& INRIA\\ \email{\texttt{\{dallago,zuppirol\}@cs.unibo.it}}}
}{
\title{Probabilistic Recursion Theory\\ and Implicit Complexity\thanks{This work is partially supported by the
  ANR project 12IS02001 PACE}}
\author{Ugo Dal Lago\footnote{Universit\`a di Bologna \& INRIA, \texttt{dallago@cs.unibo.it}} \and Sara Zuppiroli\footnote{Universit\`a di Bologna, \texttt{zuppirol@cs.unibo.it}}}
}

\maketitle

\begin{abstract}
We show that probabilistic computable functions, i.e., those functions outputting distributions
and computed by probabilistic Turing machines, can be characterized by a natural generalization of
Church and Kleene's partial recursive functions. The obtained algebra, following Leivant, can be 
restricted so as to capture the notion of polytime sampleable distributions, a key concept in 
average-case complexity and cryptography.
\end{abstract}

\section{Introduction}
Models of computation as introduced one after the other in the first
half of the last century were all designed around the assumption
that \emph{determinacy} is one of the key properties to be modeled:
given an algorithm and an input to it, the sequence of computation
steps leading to the final result is \emph{uniquely} determined by the
way an \emph{algorithm} describes the state evolution. The great
majority of the introduced models are \emph{equivalent}, in that the
classes of functions (on, say, natural numbers) they are able to
compute are the same.

The second half of the 20th century has seen the assumption above
relaxed in many different ways. Nondeterminism, as an example, has
been investigated as a way to abstract the behavior of certain classes
of algorithms, this way facilitating their study without necessarily
changing their expressive power: think about how NFAs make the task of
proving closure properties of regular languages
easier~\cite{rabinscott1959}.

A relatively recent step in this direction consists in allowing
algorithms' internal state to evolve probabilistically: the next state
is not \emph{functionally} determined by the current one, but is
obtained from it by performing a process having possibly many
outcomes, each with a certain probability. Again, probabilistically
evolving computation can be a way to abstract over determinism, but
also a way to model situations in which algorithms have access to a
source of true randomness.

Probabilistic models are nowadays more and more pervasive. Not only
are they a formidable tool when dealing with uncertainty and
incomplete information, but they sometimes are a \emph{necessity}
rather than an option, like in computational cryptography (where,
e.g., secure public key encryption schemes need to be
probabilistic~\cite{GoldwasserMicali}).  A nice way to deal
computationally with probabilistic models is to allow probabilistic
choice as a primitive when designing algorithms, this way switching
from usual, deterministic computation to a new paradigm, called
probabilistic computation.
\condinc{}{Examples of application areas in which 
probabilistic computation has proved to be useful include natural
language processing~\cite{manning1999foundations},
robotics~\cite{thrun2002robotic}, computer
vision~\cite{comaniciu2003kernel}, and machine
learning~\cite{pearl1988probabilistic}.}

But what does the presence of probabilistic choice give us in terms of
expressivity? Are we strictly more expressive than usual,
deterministic, computation? And how about efficiency: is it that
probabilistic choice permits to solve computational problems more
efficiently? These questions have been among the most central in the
theory of computation, and in particular in computational complexity,
in the last forty years (see below for more details about related
work). Roughly, while probability has been proved not to offer any
advantage in the absence of resource constraints, it is not known
whether probabilistic classes such as $\mathbf{BPP}$ or $\mathbf{ZPP}$
are different from $\mathbf{P}$.

This work goes in a somehow different direction: we want to study
probabilistic computation without necessarily \emph{reducing}
or \emph{comparing} it to deterministic computation. The central
assumption here is the following: a probabilistic algorithm computes
what we call a \emph{probabilistic function}, i.e. a function from a
discrete set (e.g. natural numbers or binary strings)
to \emph{distributions} over the same set.  What we want to do is to
study the set of those probabilistic functions which can be computed
by algorithms, possibly with resource constraints.

We give some initial results here. First of all, we provide a
characterization of computable probabilistic functions by the natural
generalization of Kleene's partial recursive functions, where among
the initial functions there is now a function corresponding to tossing
a fair coin. In the non-trivial proof of completeness for the obtained
algebra, Kleene's minimization operator is used in an unusual way,
making the usual proof strategy for Kleene's Normal Form Theorem (see,
e.g., \cite{soare1987}) useless. We later hint at how to recover the
latter by replacing minimization with a more powerful operator. We
also mention how probabilistic recursion theory offers
characterizations of concepts like the one of a computable
distribution and of a computable real number.

The second part of this paper is devoted to applying the
aforementioned recursion-theoretical framework to polynomial-time
computation. We do that by following Bellantoni and Cook's and
Leivant's works~\cite{BellantoniCook,Leivant}, in which
polynomial-time deterministic computation is characterized by a
restricted form of recursion, called \emph{predicative}
or \emph{ramified} recursion. Endowing Leivant's ramified recurrence
with a random base function, in particular, is shown to provide a
characterization of polynomial-time computable distributions, a key
notion in average-case complexity~\cite{BogdanovTrevisan06}.

\paragraph{Related Work.}
This work is rooted in classic theory of computation, and in
particular in the definition of partial computable functions as
introduced by Church and later studied by Kleene~\cite{kleene36}.
Starting from the early fifties, various forms of automata in which
probabilistic choice is available have been considered
(e.g.~\cite{Rabin63}). The inception of probabilistic choice into an
universal model of computation, namely Turing machines, is due to
Santos~\cite{santos69,santos71}, but is (essentially) already there in
an earlier work by De Leeuw and others~\cite{DeLeeuw53}. Some years
later, Gill~\cite{gill77} considered probabilistic Turing machines
with bounded complexity: his work has been the starting point of a
florid research about the interplay between computational complexity
and randomness. Among the many side effects of this research one can
of course mention modern cryptography~\cite{KatzLindell07}, in which
algorithms (e.g. encryption schemes, authentication schemes, and
adversaries for them) are very often assumed to work in probabilistic
polynomial time.

Implicit computational complexity (ICC), which studies machine-free
characterizations of complexity classed based on mathematical logic
and programming language theory, is a much younger research area.  Its
birth is traditionally made to correspond with the beginning of the
nineties, when Bellantoni and Cook~\cite{BellantoniCook} and
Leivant~\cite{Leivant} independently proposed function algebras
precisely characterizing (deterministic) polynomial time computable
functions. In the last twenty years, the area has produced many
interesting results, and complexity classes spanning from the
logarithmic space computable functions to the elementary functions
have been characterized by, e.g., function algebras, type
systems~\cite{Leivant93}, or fragments of linear
logic~\cite{girard98}.  Recently, some investigations on the interplay
between implicit complexity and probabilistic computation have started
to appear~\cite{DalLagoParisenToldin}. There is however an intrinsic
difficulty in giving \emph{implicit} characterizations of
probabilistic classes like \BPP\ or \ZPP: the latter are semantic
classes defined by imposing a polynomial bound on time, but also
appropriate bounds on the probability of error.  This makes the task
of enumerating machines computing problems in the classes much harder
and, ultimately, prevents from deriving implicit characterization of
the classes above. Again, our emphasis is different: we do not see
probabilistic algorithms as artifacts computing functions of the same
kind as the one deterministic algorithms compute, but we see
probabilistic algorithms as devices outputing distributions.

\section{Probabilistic Recursion Theory}\label{sect:prt}
\newcommand{\rand}{\mathit{rand}}
\newcommand{\CONF}[1]{\mathcal{C}_{#1}}
\newcommand{\FS}[2]{\mathcal{FC}_{\!#1}^{#2}}
\newcommand{\INIT}[2]{\mathcal{IN}_{\!#1}^{#2}}
\newcommand{\FUNC}[1]{\mathcal{IO}_{\!#1}}
\newcommand{\CT}[2]{\mathit{CT}_{\!#1}(#2)}
\newcommand{\CF}[1]{\mathcal{CF}_{\!#1}}
\newcommand{\trans}[1]{\overline{#1}}
\newcommand{\pair}{\mathit{pair}} 
\newcommand{\enc}{\mathit{enc}}
\newcommand{\PT}[1]{\mathit{PT}_{#1}}
\newcommand{\PC}[1]{\mathit{PC}_{#1}} 

In this section we provide a
characterization of the functions computed by a Probabilistic Turing
Machine (PTM) in terms of a function algebra \emph{\`a la} Kleene. We
first define {\em probabilistic recursive functions}, which are the
elements of our algebra. Next we define formally the class of
probabilistic functions computed by a PTM. Finally, we show the
equivalence of the two introduced classes. In the following, $\RRI$
is the unit interval.

Since PTMs compute probability (pseudo-)distributions, the functions
that we consider in our algebra have domain $\NN^k$ and codomain $\NN
\rightarrow \RRI$ (rather than $\NN$ as in the classic case). The idea
is that if $f(x)$ is a function which returns $r\in\RRI$ on input
$y\in\NN$, then $r$ is the probability of getting $y$ as the output
when feeding $f$ with the input $x$. We note that we could extend our
codomain from $\NN \rightarrow \RRI$ to $\NN^m \rightarrow \RRI$,
however we use $\NN \rightarrow \RRI$ in order to simplify the
presentation.
\begin{definition}[Pseudodistributions and Probabilistic Functions] \label{pseudis}
A \emph{pseudodistribution} on $\NN$ is a function
$\distrone:\fun{\NN}{\RRI}$ such that
$\sum_{\natone\in\NN}\distrone(\natone)\leq
1$. $\sum_{\natone\in\NN}\distrone(\natone)$ is often denoted as
$\sum\distrone$. Let $\distrset{\NN}$ be the set of all
pseudodistributions on $\NN$. A \emph{probabilistic function} (PF) is
a function from $\NN^k$ to $\distrset{\NN}$, where $\NN^k$ stands for
the set of $k$-tuples in $\NN$.  We use the expression $\{n_1^{p1},
\ldots, n_k^{pk} \}$ to denote the pseudodistribution $\distrone$
defined as $\distrone(n) = \sum_{n_i=n} p_i$. Observe that
$\sum\distrone=\sum_{k=1}^k p_i$.
\end{definition}
Please notice that probabilistic functions are always \emph{total}
functions, but their codomain is a set of distributions which do not
necessarily sum to $1$, but rather to a real number \emph{smaller} or
equal to $1$, this way modeling the probability of divergence. For
example, the nowhere-defined partial function $\Omega:\NN\rightharpoonup\NN$
of classic recursion theory becomes a probabilistic function which
returns the empty distributions $\emptyset$ on any input.
The first step towards defining our function algebra consists in giving a set of functions
to start from:
\begin{definition}[Basic Probabilistic Functions]\label{basprobfun}
The \emph{basic probabilistic functions} (BPFs) are as follows:
\begin{varitemize}
\item 
  The \emph{zero function} $\f{z}{\NN}{\distrset{\NN}}$ defined as:
  $z(\natone)(0)=1$ for every $\natone\in\NN$;
\item 
  The \emph{successor function} $\f{s}{\NN}{\distrset{\NN}}$ defined
  as: $s(\natone)(\natone+1) =1$ for every $n \in \NN$;
\item 
  The \emph{projection function} $\f{\Pi^n_m}{ \NN^n}
  {\distrset{\NN}}$ defined as: ${\Pi^n_m}(k_1, \cdots, k_n) (k_m)=1$
  for every positive $n,m\in \NN$ such that $1\leq m\leq n$;
\item
  The \emph{fair coin function} $\f{r}{\NN}{\distrset{\NN}}$ that
  is defined as:
  $$ 
  r(x)(y) = \left\{ 
  \begin{array}{rl} 
    1/2 & \mbox{if $y=x$} \\ 
    1/2 & \mbox{if $y=x+1$}
  \end{array} \right.
  $$
\end{varitemize}
\end{definition}
The first three BPFs are the same as the basic functions from classic
recursion theory, while $r$ is the only truly probabilistic BPF.

The next step consists in defining how PFs \emph{compose}. Function
composition of course cannot be used here, because when composing two
PFs $g$ and $f$ the codomain of $g$ does not match with the domain of
$f$. Indeed $g$ returns a distribution $\NN \rightarrow \RRI$ while
$f$ expects a natural number as input. What we have to do here is the
following.  Given an input $x\in\NN$ and an output $y\in\NN$ for the
composition ${f}\bullet{g}$, we apply the distribution $g(x)$ to any
value $z\in\NN$. This gives a probability $g(x)(z)$ which is then
multiplied by the probability that the distribution $f(z)$ associates
to the value $y\in\NN$. If we then consider the sum of the obtained
product $g(x)(z)\cdot f(z)(y)$ on all possible $z\in \NN$ we obtain
the probability of ${f}\bullet{g}$ returning $y$ when fed with $x$.
The sum is due to the fact that two different values, say
$z_1,z_2\in\NN$, which provide two different distributions $f(z_1)$
and $f(z_2)$ must both contribute to the same probability value
$f(z_1)(y)+f(z_2)(y)$ for a specific $y$. In other words, we are doing
nothing more than lifting $f$ to a function from distributions to
distributions, then composing it with $g$. Formally:
\begin{definition}[Composition] \label{Comp}
We define the \emph{composition} $\f{f \bullet
  g}{\NN}{\distrset{\NN}}$ of two functions
$\f{f}{\NN}{\distrset{\NN}}$ and $\f{g}{\NN}{\distrset{\NN}}$ as:
$$ 
\alcom{f}{g}{x}{y}= \sum_{z\in\NN}g(x)(z)\cdot f(z)(y).
$$
\end{definition}
The previous definition can be generalized to functions taking more
than one parameter in the expected way:
\begin{definition} [Generalized Composition] \label{GComp}
We define the \emph{generalized composition} of functions
$\f{f}{\NN^n}{\distrset{\NN}}$,
$\f{g_1}{\NN^k}{\distrset{\NN}},\ldots,\f{g_n}{\NN^k}{\distrset{\NN}}$
as the function $\f{f \odot (g_1,\ldots,g_n)}{\NN^k}{\distrset{\NN}}$
defined as follows:
$$
\alcomn{f}{(g_1,\ldots,g_n)}{\V{x}}{y}=\sum_{z_1,\ldots,z_n\in\NN}\left(f(z_1,\ldots,z_n)(y)\cdot\prod_{1\leq i\leq n} g_i(\V{x})(z_i)\right).
$$
\end{definition}
With a slight abuse of notation, we can treat probabilistic functions
as ordinary functions when forming expressions. Suppose, as an
example, that $x\in\NN$ and that $\f{f}{\NN^3}{\distrset{\NN}}$,
$\f{g}{\NN}{\distrset{\NN}}$, $\f{h}{\NN}{\distrset{\NN}}$. Then the
expression $f(g(x),x,h(x))$ stands for the distribution in
$\distrset{\NN}$ defined as follows:
$\alcomnf{f}{(g,\mathit{id},h)}{x}$, where $\mathit{id}=\Pi^1_1$ is
the identity PF.

The way we have defined probabilistic functions and their composition
is reminiscent of, and indeed inspired by, the way one defines the
Kleisli category for the Giry monad, starting from the category of
partial functions on sets. This categorical way of seeing the problem
can help a lot in finding the right definition, but by itself is not
adequate to proving the existence of a correspondence with machines
like the one we want to give here.

Primitive recursion is defined as in Kleene's algebra, provided that
one uses composition as previously defined:
\begin{definition} [Primitive Recursion] \label{prec}
Given functions $\f{g}{\NN^{k+2}} {\distrset{\NN}}$, and
$\f{f}{\NN^k}{\distrset{\NN}}$, the function $\f{h}
{\NN^{k+1}}{\distrset{\NN}}$ defined as
$$
h(\V{x},0)=f(\V{x});\qquad
h(\V{x},y+1)=g(\V{x},y,h(\V{x},y)); 
$$ 
is said to be defined \emph{by primitive recursion} from $f$ and
$g$, and is denoted as $\primrec{f}{g}$.
\end{definition} 

We now turn our attention to the minimization operator which, as in
the deterministic case, is needed in order to obtain the full
expressive power of (P)TMs. The definition of this operator is in our
case delicate and requires some explanation. Recall that, in the
classic case, the minimization operator allows from a partial function
$f:\NN^{k+1}\rightharpoonup\NN$, to define another partial function,
call it $\mu f$, which computes from $\V{x}\in\NN^{k}$ the least value
of $y$ such that $f(\V{x},y)$ is equal to $0$, if such a value exists
(and is undefined otherwise).  In our case, again, we are concerned
with distributions, hence we cannot simply consider the least value on
which $f$ returns $0$, since functions return $0$ \emph{with a certain
  probability}.  The idea is then to define the minimization $\mu f$
as a function which, given an input $\V{x}\in\NN^k$, returns a
distribution associating to each natural $y$ the probability that the
result of $f(\V{x},y)$ is $0$ \emph{and} the result of $f(\V{x},z)$ is
positive for every $z<y$. Formally:
\begin{definition}[Minimization] \label{fmin} 
Given a PF $\f{f}{\NN^{k+1}}{\distrset{\NN}} $, we define another PF 
$\f{\mu f} {\NN^k} {\distrset{\NN}}$ as follows:
$$
\mu f(\V{x})(y)=f(\V{x},y)(0) \cdot (\prod_{z<y}( \sum_{k>0} f(\V{x},z)(k))).
$$
\end{definition} 
We are finally able to define the class of functions we are interested
in as follows.
\begin{definition} [Probabilistic Recursive Functions] \label{Pu}
The class $\PrR$ of \emph{probabilistic recursive functions} is the
smallest class of probabilistic functions that contains the BPFs
(Definition \ref{basprobfun}) and is closed under the operation of
General Composition (Definition \ref{GComp}), Primitive Recursion
(Definition \ref{prec}) and Minimization (Definition \ref{fmin}).
\end{definition}
\condinc{}{It is easy to show that $\PrR$ includes all partial
  recursive functions.  This can be done by first defining an extended
  Recursive Function as follows.
\begin{definition} [Extended Recursive Functions] \label{ProbKleene}
For every partial recursive function  $\f{f}{\NN^k}{\NN}$ we define as the extended function $\f{p_f}{\NN^k}{\distrset{\NN}}$ 
as follows:
$$
p_f(\V{x})(y)= \left\{ \begin{array}{rl} 
 1& $if $y = f(\V{x}) \\
 0 & otherwise
  \end{array} \right.
$$
\end{definition}
\begin{proposition}\label{Precfunc}
If $f$ is a Partial Recursive function then we define $p_f$ as defined above is in  $\PrR$.
\end{proposition}
\begin{proof}
The proof goes by induction on the structure of $f$ as a partial recursive function. 
\begin{itemize}
\item $f$ is the zero function, so $f : \NN  \rightarrow \NN$  defined as:  $f(x) = 0$ for every $x  \in \NN$.
Thus $p_f$ is in $\PrR$ because $p_f=z$ 
\item $f$ is the successor function $s : \NN\rightarrow \NN$  defined as:  $s(x) = x + 1$ for every $x \in \NN$.
Thus $p_f$ is in $\PrR$ because $p_f=s$ 
\item $f$ is the projection function. $f^n_m: \NN^n \rightarrow \NN$ defined as:  $f^n_m(x_1, \cdots, x_n) = x_m$
for every positive $n \in \NN$ and for all $m\in \NN$, such that $1\leq m\leq n$
Thus $p_f$ is in $\PrR$ because $p_f=\Pi^n_m$ 

\item $f$ is defined by composition from $h, g_1, \cdots, g_n$ as:
$$
f(\V{x}) = h(g_1(\V{x}), \cdots, g_n(\V{x}))
$$
where $\f{h} {\NN^n} {\NN}$ and
$\f{g_i} {\NN^k} {\NN}$ for every $1\leq i\leq n$ are partial recursive functions.
by definition of $f(\V{x})$ we have
$$
p_f(\V{x})(y)= \left\{ \begin{array}{rl} 
 1& $if $y =  h(g_1(\V{x}), \cdots, g_n(\V{x})) \\
 0 & otherwise
  \end{array} \right.
$$
We see that $h, g_1, \cdots, g_n$ are all partial recursive functions. 
So we have by definition of $p_f$ that
$$
p_{g_1}(\V{x})(y)= \left\{ \begin{array}{rl} 
 1& $if $y =g_1(\V{x}) \\
 0 & otherwise
  \end{array} \right.
 $$
 $$
 \vdots
 $$
$$
p_{g_n}(\V{x})(y)= \left\{ \begin{array}{rl} 
 1& $if $y =g_n(\V{x}) \\
 0 & otherwise
  \end{array} \right.
$$
$$
p_{h}(\V{z})(y)= \left\{ \begin{array}{rl} 
 1& $if $y =h(\V{z}) \\
 0 & otherwise
  \end{array} \right.
$$

By hypothesis we observe that $p_{g_1}, \cdots, p_{g_n}, p_h \in \PrR$ and 

\begin{align*}
\alcomn{p_{h}}{(p_{g_1},\ldots,p_{g_n})}{\V{x}}{y}&=\sum_{z_1,\ldots,z_n\in\NN} p_h(z_1,\ldots,z_n)(y)\cdot(\prod_{1\leq i\leq n} p_{g_i}(\V{x})(z_i))\\
&= \sum_{z_1,\ldots,z_n\in\NN} p_h(z_1,\ldots,z_n)(y)\cdot (\prod_{z_i=g_i(\V{x})} 1) \\
&=\sum_{y=h(z_1, \cdots, z_n)} 1 \cdot( \prod_{z_i=g_i(\V{x})} 1)\\
&=\sum_{y=h(g_1(\V{x}),\cdots, g_n(\V{x}))} 1\\
& =  p_f(\V{x})(y)
\end {align*}
Thus $p_f$ is in $\PrR$ because $p_f=\alcomn{p_{h}}{(p_{g_1},\ldots,p_{g_n})}{\V{x}}{y}$ .

\item $f$  is defined by primitive recursion so 
$\f{f} {\NN^k\times \NN}{\NN}$ defined as:
$$
f(\V{x},0) = h(\V{x})
$$
$$
f(\V{x},n+1)=g(\V{x},n, f(\V{x},n))
$$
where $\f{g} {\NN^k\times\NN\times\NN} {\NN}$ and
$\f{h} {\NN^k} {\NN}$ are partial recursive functions.

$p_f$ is defined as:
$$
p_f(\V{x},n)(z_n)= \left\{ \begin{array}{rl} 
 1& $if $ z_n =\primrec{h}{g} \\
 0 & otherwise
  \end{array} \right.
$$

We see that $h, g$ are all partial recursive functions. 
So we have by definition of $p_f$ that

$$
p_{h}(\V{x})(z_0)= \left\{ \begin{array}{rl} 
 1& $if $y =h(\V{x}) \\
 0 & otherwise
  \end{array} \right.
$$
$$
p_{g}(\V{x},n, z_{n})(z_{n+1})= \left\{ \begin{array}{rl} 
 1& $if $z_{n+1} =g(\V{x},n,f(\V{x},n)) \\
 0 & otherwise
  \end{array} \right.
$$


By hypothesis we observe that $p_g,  p_h \in \PrR$. Now 
if $n=0$ then $p_f(\V{x},0)= p_{h}(\V{x})$ and if $n>0$ then $p_f(\V{x},n+1)= p_{g}(\V{x},n, z_{n})$.
We observe that
\begin{align*}
\alcomn{p_{g}}{(\id, p_f)}{\V{x},n}{z_{n+1}}&=\sum_{x_1,\ldots,z_k, n, z_{n}\in\NN} p_g(x_1,\ldots,x_k,n,z_{n})(z_{n+1})\cdot(\prod_{1\leq i\leq k+1} \id(\V{x},n)(\V{x},n) \cdot f(\V{x},n)(z_n))\\
&= \sum_{x_1,\ldots,z_k, n, z_{n}\in\NN} p_g(x_1,\ldots,x_k,n,z_{n})(z_{n+1})\cdot (\prod_{\V{x}=\V{x} , n=n , z_{n}=f(\V{x},n)} 1  ) \\
&=\sum_{z_{n+1}=g(x_1,\ldots,x_k, n, z_{n})} 1 \cdot (\prod_{\V{x}=\V{x} , n=n , z_{n}=f(\V{x},n)} 1  )\\
&=\sum_{z_{n+1}=g(\V{x},n,f(\V{x},n))} 1 \\
& =  p_f(\V{x},n+1)(z_{n+1})
\end {align*}
Thus $p_f$ is in $\PrR$ because $p_f=\primrec{p_h}{p_g} $.

\item $f$  is defined by minimization so:
$$
f(\V{x}) = \mu \; y \; (g(\V{x},y)=0) 
$$
$p_f$ is defined as:
$$
p_f(\V{x})(z)= \left\{ \begin{array}{rl} 
 1& $if $z =f(\V{x})  \\
 0 & otherwise
  \end{array} \right.
$$
by definition of $f(\V{x})$ we have:
$$
p_f(\V{x})(z)= \left\{ \begin{array}{rl} 
 1& $if $z = \mu \; y \; (g(\V{x},y)=0)  \\
 0 & otherwise
  \end{array} \right.
$$
We know that $g$ is a recursive function, so we have that:
$$
p_g(\V{x},z)(k)= \left\{ \begin{array}{rl} 
 1& $if $k = g(\V{x},z)  \\
 0 & otherwise
  \end{array} \right.
$$ 

By hypothesis $p_g \in \PrR$.
We observe that:

\begin{align*}
\mu \; p_g(\V{x})(z)&=p_g(\V{x},z)(0) \cdot (\prod_{n< z}( \sum_{k>0} p_g(\V{x},n)(k)))\\
&=  p_g (\V{x},z)(0) \cdot (\prod_{n< z}( \sum_{k>0, k= g(\V{x},n)} 1)) \\
&= p_g (\V{x},z)(0) \cdot (\prod_{n< z, k>0, k= g(\V{x},n)} 1) \\
&= \left\{ \begin{array}{rl} 
 1& $if z is the minimal values such that  $ g(\V{x},z) = 0 $  and for all $ n< z $ $ g(\V{x},z) > 0\\
 0 & otherwise
  \end{array} \right.  \\ 
& = \left\{ \begin{array}{rl} 
 1& $if $ z= \mu \; y \; (g(\V{x},y)=0)  \\
 0 & otherwise
  \end{array} \right. \\ 
&=p_f(\V{x})(z)
\end {align*}
Thus $p_f$ is in $\PrR$ because $p_f=\mu \; p_g$
\end{itemize}

%
%
%
\end{proof}

}
It is easy to show that $\PrR$ includes all partial recursive functions, seen
as probabilistic functions: first, for every partial function 
$f:\NN^k\rightharpoonup\NN$, define $\f{p_f}{\NN^k}{\distrset{\NN}}$ 
by stipulating that $p_f(\V{x})(y)=1$ whenever $y = f(\V{x})$, and $p_f(\V{x})(y)=0$ otherwise;
then, $p_f\in\PrR$ whenever $f$ is partial recursive.

\condinc {
\begin{example}
  The following are examples of probabilistic recursive functions:
  \begin{varitemize}
 \item
   The \emph{identity function} $\f{\id}{\NN}{\distrset{\NN}}$, defined as 
   $\id(x)(x)=1$.
   For all $x, y \in\NN$ we have that 
   $$
   \id(x)(y) = 
   \left\{ \begin{array}{ll} 
       1 & \mbox{if $y=x$} \\
       0 & \mbox{otherwise}
   \end{array} \right.
   $$
   as a consequence $id=\Pi^1_1$, and, since the 
   latter is a BPF (Definition \ref{basprobfun}) $id$ is in $\PrR$.
  \item
    \newcommand{\add}{\mathit{add}} 
    The probabilistic funtion $\f{\rand}{\NN}{\distrset{\NN}}$ such
    that for every $x\in\NN$, $\rand(x)(0)=\frac{1}{2}$ and
    $\rand(x)(1)=\frac{1}{2}$ can be easily shown to be recursive,
    since $\rand=r\odot z$ (and we know that both $r$ and $z$ are
    BPF). Actually, $\rand$ could itself be taken as the only genuinely
    probabilistic BPF, i.e., $r$ can be constructed from $\rand$ and
    the other BPF by composition.  We proceed by defining
    $\f{g}{\NN^3}{\distrset{\NN}}$ as follow:
    $$
    g(x_1,x_2, z)(y) = 
    \left\{ \begin{array}{ll} 
      1 & \mbox{if $y=z+1$} \\
      0 & \mbox{otherwise}
    \end{array} \right.
    $$ 
    $g$ is in $\PrR$ because $g=s\odot(\Pi^3_3)$.  Now we observe that
    the function $\add$ defined by $\add(x,0)=\id(x)$ and
    $\add(x_1,x_2+1)=g(x_1,x_2,\add(x_1,x_2))$ is a probabilistic
    recursive function, since it can be obtained from basic functions
    using composition and primitive recursion.  We can conclude by
    just observing that $r=\add\odot(\Pi^1_1,\rand)$.
  \item
    All functions we have proved recursive so far have the property
    that the returned distribution is \emph{finite} for any input.
    Indeed, this is true for every probabilistic \emph{primitive}
    recursive function, since minimization is the only way to break
    this form of finiteness.  Consider the function
    $\f{f}{\NN}{\distrset{\NN}}$ defined as
    $f(x)(y)=\frac{1}{2^{y-x+1}}$ if $y\geq x$, and $f(x)(y)=0$
    otherwise.  We define another function
    $\f{h}{\NN}{\distrset{\NN}}$ by stipulating that
    $h(x)(y)=\frac{1}{2^{y+1}}$ for every $x,y\in\NN$. $h$ is a
    probabilistic recursive function; indeed consider the function
    $\f{k}{\NN^2}{\distrset{\NN}}$ defined as $\rand\odot\Pi^2_1$ and build
    $\mu\;k$. By definition,
    \begin{equation}\label{equ:muk}
    (\mu\;k)(x)(y)=k(x, y)(0)\cdot(\prod_{z<y}( \sum_{q>0} k(x, z)(q))). 
    \end{equation}
    Then observe that $(\mu\;k)(x)(y)=\frac{1}{2^{y+1}}$: by
    (\ref{equ:muk}), $(\mu\;k)(x)(y)$ unfolds into a product of exactly
    $y+1$ copies of $\frac{1}{2}$, each ``coming from the flip of a
    distinct coin''. Hence, $h=\mu\;k$.  Then we observe that
    $$
    (\add\odot(\mu\;k,\id))(x)(y)=\sum_{z_1,z_2}\add(z_1,z_2)(y) \cdot ((\mu \;k)(x)(z_1)\cdot\id(x)(z_2)).
    $$ 
    But please notice that $\id(x)(z_2)=1$
    only when $z_2=x$ (and in the other cases $id(x)(z_2) = 0$),
    $(\mu\;k)(x)(z_1)=\frac{1}{2^{z_1+1}}$, and $\add(z_1,z_2)(y)=1$ only when
    $z_1+z_2=y$ (and in the other cases, $\add(z_1,z_2)(y)=0$).
    This implies that  the term in the sum is different from $0$ only when 
    $z_2=x$ and $z_1+z_2=y$, namely when $z_1=y-z_2=y-x$, and in that case its value is
    $\frac{1}{2^{y-x+1}}$. Thus, we can claim that
    $f=(\add\odot(\mu\;k,\id))$.
  \end{varitemize}
\end{example}
}{
\begin{example}
  The following are examples of probabilistic recursive functions:
  \begin{varitemize}
  \item
   The \emph{identity function} $\f{\id}{\NN}{\distrset{\NN}}$, defined as 
   $\id(x)(x)=1$.
   For all $x, y \in\NN$ we have that 
   $$
   id(x)(y) = \left\{ \begin{array}{rl} 
       1 & $if $y = x \\
       0 & otherwise
     \end{array} \right.
   $$
   on the other hand for every $x,y \in \NN$ we have as a consequence, $id=\Pi^1_1$, and, since the 
   latter is a basic function (Definition \ref{basprobfun}) $id$ is in $\PrR$.
  \item
    The function $\f{f}{\NN}{\distrset{\NN}}$ defined by
    $f(x)(x)=\frac{1}{2}$ and $f(x)(x+1)=\frac{1}{2}$.
    
    We define $\f{id}{\NN}{\distrset{\NN}}$ as follow:
    $$
    id(x)(y) = \left\{ \begin{array}{rl} 
        1 & $if $y = x \\
        0 & otherwise
      \end{array} \right.
    $$
    We define $\f{g}{\NN^3}{\distrset{\NN}}$ as follow:
    $$
    g(x_1,x_2, z)(y) = \left\{ \begin{array}{rl} 
        1 & $if $y = z+1 \\
        0 & otherwise
      \end{array} \right.
    $$
    $g$ is in $\PrR$ because $g=s\odot(\Pi^3_3)$ 
    
    Finally the function $add$ satisfies the follow equations:
    \begin{align*}
      add(x,0)&=h(x);\\
      add(x_1,x_2+1)&=g(x_1,x_2,add(x_1,x_2))
    \end{align*}
    
    so we proved that $add$ is constructed by Primitive Recursion 
    operation and it is a probabilistic recursive function.
    Now we observe that:
    $$
    f=add\odot(\Pi^1_1, rand)
    $$
    that is a function defined using the operation of  General Composition (Definition \ref{GComp}).
    
  \item  
    The function $\f{f}{\NN}{\distrset{\NN}}$ defined by
    $$
    f(x)(y)=
    \left\{
      \begin{array}{ll}
        \frac{1}{2^{y-x}} & \mbox{if $y>x$}\\
        0 & \mbox{otherwise}
      \end{array}  \right.
    $$
   We define $\f{h}{\NN}{\distrset{\NN}}$ as follow:
    $$
    h(x) = \left\{ \begin{array}{rl} 
        1/2^y & $if $y \ge 1 \\
        0 & otherwise
      \end{array} \right.
    $$
    that is a probabilistic recursive function because
    $$
    \mu \;  rand(x)(y)= rand(x, y)(0)  \cdot (\prod_{y< z}( \sum_{k>0} rand(x, z)(k))) 
    $$
    thus $rand(x, y)(0) = 1/2$ and $\sum_{k>0} rand(x, z)(k) = \sum_{k=1,2,\cdots} rand(x, z)(k)= rand(x, z)(1)+rand(x,z)(2)+\cdots=1/2+0+\cdots =1/2$ for definition of rand.
    $\prod_{y< z}1/2=1/2^{y-1}$.
    So
    $$
    \mu \;  rand(x)(y)=  \left\{ \begin{array}{rl} 
        1/2^y & $if $y \ge 1 \\
        0 & otherwise
      \end{array} \right.
    $$
    Then we observe that: 
    $$
    g(x)(y)=  add \odot ( \mu \;  rand, id) (x) (y)
    $$
    So
    $$
    add \odot ( \mu \;  rand, \id) (x) (y)=   \sum_{x_1,x_2}  add(x_1,x_2)(y) \cdot (\mu \; rand(x)(x_1)*id(x)(x_2)) = 1/2^{y-x}
    $$ 
    because the function $id(x)(x) = 1$
    and so $x_2=x$
    and $\mu \; rand (x) (x_1) = 1/2^{x_1}$ if $x_1>0$.
    So $x_1>0$.
    Now this is true if and only if 
    $x_2=x$ and $x+x_1=y$ and finally $x_1=y-x$.
  \end{varitemize}
\end{example}
}
\subsection{Probabilistic Turing Machines and Computable Functions}
In this section we introduce computable functions as those
probabilistic functions which can be computed by Probabilistic Turing
Machines.  As previously mentioned, probabilistic computation devices
have received a wide interest in computer science already in the
fifties~\cite{DeLeeuw53} and early sixties~\cite{Rabin63}. A natural
question which arose was then to see what happened if random elements
were allowed in a Turing machine. This question led to several
formalizations of probabilistic Turing machines (PTMs in the
following)~\cite{DeLeeuw53,santos69} --- which, essentially, are
Turing machines which have the ability to flip coins in order to make
random decisions --- and to several results concerning the
computational complexity of problems when solved by
PTMs~\cite{gill77}.
  
Following \cite{gill77}, a Probabilistic Turing Machine (PTM) $M$ can
be seen as a Turing Machine with two transition functions $\delta_0,
\delta_1$. At each computation step, either $\transone_0$ or
$\transone_1$ can be applied, each with probability $1/2$.  Then, in a
way analogous to the deterministic case, we can define a notion of a
(initial, final) configuration for a PTM $\ptmone$. In the following,
$\Sigma_b$ denotes the set of possible symbols on the tape, including
a blank symbol; $Q$ denotes the set of states; $Q_f\subseteq Q$
denotes the set of final states and $q_s\in Q$ denotes the initial
state.
\begin{definition}[Probabilistic Turing Machine]\label{PTM}
A Probabilistic Turing Machine (PTM) is a Turing machine endowed with
two transition functions $\transone_0, \transone_1$. At each
computation step the transition function $\transone_0$ can be applied
with probability $1/2$ and the transition $\transone_1$ can be applied
with probability $1/2$.
\end{definition}
\begin{definition}[Configuration of a PTM] \label{ConfigPTM}
Let $\ptmone$ be a PTM.  We define a PTM \emph{configuration} as a
4-tuple
$\conf{\strone}{\symone}{\strtwo}{\stone}\in\bk{\alpone}^*\times\bk{\alpone}\times\bk{\alpone}^*
\times\ssone $ such that:
\begin{varitemize}
\item 
  The first component, $\strone\in\bk{\alpone}^*$, is the portion of
  the tape lying on the left of the head.
\item 
  The second component, $\symone\in\bk{\alpone}$, is the symbol the
  head is reading.
\item 
  The third component, $\strtwo\in\bk{\alpone}^*$, is the portion of
  the tape lying on the right of the head.
\item 
  The fourth component, $\stone\in\ssone$ is the current state.
\end{varitemize}
Moreover we define the set of all configurations as $\CONF{M} =
\bk{\alpone}^*\times\bk{\alpone}\times\bk{\alpone}^* \times\ssone $.
\end{definition}

\begin{definition}[Initial and Final Configurations of a PTM] \label{initialConfigPTM}
Let $\ptmone$ be a PTM.  We define the \emph{initial configuration} of
$\ptmone$ for the string $\strone$ as the configuration in the form
$\langle\varepsilon, a , v, q_s\rangle\in
\Sigma_b^*\times\Sigma_b\times\bk{\alpone}^* \times\ssone $ such that
$\strone=a\cdot v$ and the fourth component, $q_s\in \ssone$, is the
initial state.  We denote it with $\INIT{M}{\strone}$. Similarly, we
define a \emph{final configuration} of $\ptmone$ for $\strone$ as a
configuration $\langle \strone,\_,\varepsilon,q_f\rangle\in
\Sigma_b^*\times\Sigma_b\times\Sigma_b^*\times Q_f$.  The set of all
such final configurations for a PTM $M$ is denoted by
$\FS{M}{\strone}$.
\end{definition}
For a function $\f{T}{\NN}{\NN}$, we say that a PTM $M$ \emph{runs in time bounded by $T$} if for any input $x$, 
$M$ halts on input $x$ within $T(|x|)$ steps \emph{independently} of the random choices it makes. Thus, $M$
\emph{works in polynomial time} if it runs in time bounded by $P$, where $P$ is any polynomial.

Intuitively, the function computed by a PTM $\ptmone$ associates to
each input $\strone$, a (pseudo)-distribution which indicates the
probability of reaching a final configuration of $M$ from
$\INIT{M}{\strone}$. It is worth noticing that, differently from the
deterministic case, since in a PTM the same configuration can be
obtained by different computations, the probability of reaching a
given final configuration is the \emph{sum} of the probabilities of
reaching the configuration along all computation paths, of which there
can be (even infinitely) many.  It is thus convenient to define the
function computed by a PTM through a fixpoint construction, as
follows.  First, we can define a partial order on the string
distributions as follows.  
\condinc{
\begin{definition}\label{def:strdistr}
A \emph{string pseudodistribution} on $\Sigma^*$ is a function
$\distrstringone:\fun{\alpone^*}{\RRI}$ such that
$\sum_{s\in\Sigma^*}\distrstringone(s)\leq 1$. $\distrset{\alpone^*}$
denotes the set of all string pseudodistributions on $\alpone^*$. The
relation
$\sqsubseteq_{\distrset{\alpone^*}}\subseteq\distrset{\alpone^*}\times\distrset{\alpone^*}$
is defined as the pointwise extension of the usual partial order on
$\RR$ 
\condinc{.}  {: for
  $\distrstringone,\distrstringtwo\in\distrset{\alpone^*}$ by imposing
  that
  $\distrstringone\sqsubseteq_{\distrset{\alpone^*}}\distrstringtwo$
  if and only if, for all $\strone\in\alpone^*$, it holds
  $\distrstringone(s)\leq\distrstringtwo(s)$.}
\end{definition}
It is easy to show that the relation $\sqsubseteq_{\distrset{\alpone^*}}$ from Definition~\ref{def:strdistr} is a partial order.
}
{
\begin{definition}\label{def:strdistr}
A \emph{string distribution} on $\Sigma^*$ is a function $\distrstringone:\fun{\alpone^*}{\RRI}$ such that
$\sum_{s\in\Sigma^*}\distrstringone(s)\leq 1$. $\distrset{\alpone^*}$ denotes the set of all
string distributions on $\Sigma^*$. 
\end{definition}
Next we can define a partial order on string distributions by a pointwise extension of the usual order on $\RR$:
\begin{definition}\label{relord1}
The relation $\sqsubseteq_{\distrset{\alpone^*}} \subseteq \distrset{\alpone^*}\times\distrset{\alpone^*}$ is defined 
by stipulating that $A \sqsubseteq_{\distrset{\alpone^*}} B$ if and only if, for all $s \in \Sigma^*$,  $A(s) \leq B(s)$.
\end{definition}
The proof of the following is immediate.
\begin{proposition} \label{relord4}
The structure $(\distrset{\alpone^*},\sqsubseteq_{\distrset{\alpone^*}})$ is a POSET.
\end{proposition}
}
\condinc{
Next, we can define the domain $\confeval$ of those functions computed by a PTM $M$ from a given
configuration, i.e., the set of those functions $f$ such that $\f{f}{\CONF{M}}{\distrset{\alpone^*}}$.
Inheriting the structure from $\distrset{\alpone^*}$, we can obtain a poset $(\confeval,\sqsubseteq_{\confeval})$, again 
by defining $\sqsubseteq_{\confeval}$ pointwise.
}
{
Now we can define the domain $\confeval$ of those functions computed by a PTM $M$ from a given
configuration\footnote{Of course $\confeval$  is a proper superset of the functions computed by PTMs.}. 
This set is defined as follows and will be used as the domain of the functional whose least fixpoint gives the function computed by a PTM.
\begin{definition}\label{confeval}
The set $\confeval$ is defined as $\{f | \f{f}{\CONF{M}}{\distrset{\alpone^*}} \}$
\end{definition}
Inheriting the structure on $\distrset{\alpone^*}$ we can define a partial order on $\confeval$ as follows.
\begin{definition}\label{relord2}
The relation $\sqsubseteq_{\confeval} \subseteq \confeval \times \confeval$ is defined for
$A, B \in \confeval$ $A \sqsubseteq_{\confeval} B$ if and only if, for all $c \in \CONF{M} $, $A(c) \sqsubseteq_{\distrset{\alpone^*}} B(c)$
\end{definition}
Also the proof of the following is immediate.
\begin{proposition}\label{relord3}
The structure $(\confeval,\sqsubseteq_{\confeval})$ is a POSET.
\end{proposition}
}
%
%
%
\condinc{
Moreover, it is also easy to show that the two introduced posets are $\omega\mathbf{CPO}$s.
}
{

%
%
%
%
%
Given a POSET, the notions of least upper bound, denoted by $\bigsqcup$, and of an ascending chain are defined as usual. Next, 
the bottom elements are defined as follows.
\begin{lemma}\label{LUBD}
Let $d_\perp: \Sigma^* \rightarrow \RRI$ be defined by stipulating that
$d_\perp(s)=0$ for all $s\in\Sigma^*$. Then, $d_\perp$ is the bottom element of the poset 
$(\distrset{\alpone^*},  \sqsubseteq_{\distrset{\alpone^*}})$.
\end{lemma}
\begin{lemma}\label{LUBC}
Let $b_\perp: \CONF{M} \rightarrow \distrset{\alpone^*}$ be defined by stipulating that
$b_\perp(c)= d_\perp$ for all $c \in \CONF{M}$. Then, $b_\perp$ is the bottom element of 
the poset $(\confeval,\sqsubseteq_{\confeval})$.
\end{lemma}
Now, it is time prove that the posets at hand are also $\omega$-complete:
\begin{proposition}\label{wcpodis}
The POSET $(\distrset{\alpone^*}, \sqsubseteq_{\distrset{\alpone^*}})$ is a $\omega$CPO.
\end{proposition}
\begin{proof}
We need to prove that  for each chain 
$$
c_1\sqsubseteq_{\distrset{\alpone^*}}c_2\sqsubseteq_{\distrset{\alpone^*}}c_3  \ldots
$$  
the least upper bound $\bigsqcup_{i} c_i$ exists. First note that since 
$\sum_{s\in\Sigma^*} c_i(s) \leq 1$, from definition of $ \sqsubseteq_{\distrset{\alpone^*}}$ it follows that, 
for each $s\in\Sigma^*$,  $c_1(s) \leq c_2(s)\leq \ldots \leq 1$ holds.  This implies that, 
for each $s\in\Sigma^*$, the limit $\lim_{i\rightarrow \infty } c_i(s)$ exists. Hence, defining 
$c_{\mathit{LIM}}$ as the distribution such that $c_{\mathit{LIM}}(s) = \lim_{i\rightarrow \infty } c_i(s)$,
we have that $c_{\mathit{LIM}}=\bigsqcup_{i} c_i$. Indeed, $c_{\mathit{LIM}}\sqsupseteq_{\distrset{\alpone^*}}c_i$,
and any upper bounds of the family $\{c_i\}_{i\in\NN}$ is clearly greater or equal to $c_{\mathit{LIM}}$.
\end{proof}

\begin{proposition}\label{wcpocon}
The POSET $(\confeval,  \sqsubseteq_{\confeval})$ is a $\omega$CPO.
\end{proposition}

\begin{proof}
Analogous to the previous one.
\end{proof}
}

We can now define a functional $F_M$ on $\confeval$ which will be used to define the function computed by $M$ 
via a fixpoint construction. Intuitively, the application of the functional $F_M$ describes \emph{one} computation step. 
Formally:
\begin{definition}\label{def:functionalPTM}
Given a PTM $\ptmone$, we define a functional $\f{F_M}{\confeval}{\confeval} $ as:
$$
F_M(f)(C)= \left\{ \begin{array}{ll} 
      \{\strone^1\}& \mbox{ if } C\in\FS{M}{\strone};\\
      \frac{1}{2}f(\delta_0(C))+\frac{1}{2}f(\delta_1(C)) & \mbox{ otherwise}.
    \end{array} \right.
$$
\end{definition}
\condinc{
One can show that the functional $F_M$ from Definition \ref{def:functionalPTM} is continuous on $\confeval$. 
A classic fixpoint theorem ensures that $F_M$ has a least fixpoint. 
Such a least fixpoint is, once composed with a function returning $\INIT{M}{\strone}$ from
$\strone$, the \emph{function computed by the machine $M$}, which is denoted as 
$\FUNC{M}:\alpone^*\rightarrow\distrset{\alpone^*}$. The set of those functions which
can be computed by any PTM is denoted as $\PrC$, while $\PrPC$ is the set of probabilistic
functions which can be computed by a PTM working in \emph{polynomial} time.
}
{
The following proposition is needed in order to apply the usual fix point result.
\begin{proposition}\label{cont}
The functional $F_M$ is continuos.
\end{proposition}
\begin{proof}
We prove that
$$
F_M (\bigsqcup_{i\in\NN} f_i) =\bigsqcup_{i\in\NN} (F_M (f_i)),
$$
or, saying another way, that for every configuration $C$,
$$
F_M (\bigsqcup_{i\in\NN} f_i)(C) =\bigsqcup_{i\in\NN} (F_M (f_i))(C).
$$
Now, notice that for every $C$,
$$
F_M(\bigsqcup_{i\in\NN} f_i)(C) =  
\left\{ 
\begin{array}{ll} 
    \{\strone^1\} & \mbox{if $C \in \FS{M}{\strone}$} \\
    \frac{1}{2}((\bigsqcup_{i\in\NN} f_i)(C_1))+\frac{1}{2}((\bigsqcup_{i\in\NN} f_i)(C_2)) & \mbox{if  $C \rightarrow C_1, C_2$} \\
\end{array} \right.
$$
and, similarly, that:
$$
\bigsqcup_{i\in\NN} (F_M (f_i))(C) = \bigsqcup_{i\in\NN} 
\left\{ \begin{array}{ll} 
    \{\strone^1\}& \mbox{if $C\in\FS{M}{\strone}$}\\
    \frac{1}{2}f_i(C_1)+\frac{1}{2}f_i(C_2) & \mbox{if $C \rightarrow C_1,C_2$} \\
\end{array} \right.
$$
Now, given any $C$, we distinguish two cases:
\begin{varitemize}
\item
  If $C\in\FS{M}{\strone}$, then
  $$
  F_M(\bigsqcup_{i\in\NN} f_i)(C)=\{s^1\}=\bigsqcup_{i\in\NN}\{s^1\}=\bigsqcup_{i\in\NN} (F_M (f_i))(C).
  $$
\item
  If $C \rightarrow C_1,C_2$, then
  \begin{align*}
    F_M(\bigsqcup_{i\in\NN} f_i)(C)&=\frac{1}{2}((\bigsqcup_{i\in\NN} f_i)(C_1))+\frac{1}{2}((\bigsqcup_{i\in\NN} f_i)(C_2))\\
    &=\frac{1}{2}(\bigsqcup_{i\in\NN} f_i(C_1))+\frac{1}{2}(\bigsqcup_{i\in\NN} f_i(C_2))\\
    &=\bigsqcup_{i\in\NN}\frac{1}{2}f_i(C_1)+\bigsqcup_{i\in\NN}\frac{1}{2}f_i(C_2)=\bigsqcup_{i\in\NN}(\frac{1}{2}f_i(C_1)+\frac{1}{2}f_i(C_2))\\
    &=\bigsqcup_{i\in\NN} (F_M (f_i))(C).
  \end{align*}
\end{varitemize}
This concludes the proof.
\end{proof}

\begin{theorem}\label{mon}
The functional defined in \ref{def:functionalPTM} has a least fix point which 
is equal to $\bigsqcup_{n\geq0} F_M^n (b_\perp)$.
\end{theorem}
\begin{proof}
Immediate from the well-known fix point theorem for continuous maps on a $\omega$CPO.
\end{proof}

Such a least fixpoint is, once composed with a function returning $\INIT{M}{\strone}$ from
$\strone$, the \emph{function computed by the machine $M$}, which is denoted as 
$\FUNC{M}:\alpone^*\rightarrow\distrset{\alpone^*}$. The set of those functions which
can be computed by any PTMs is denoted as $\PrC$.

}
The notion of a computable probabilistic function subsumes other key notions in probabilistic and real-number computation. 
As an example, \emph{computable distributions} can be characterized as those distributions
on $\alpone^*$ which can be obtained as the result of a function in $\PrC$ on a \emph{fixed} input. Analogously,
\emph{computable real numbers} from the unit interval $[0,1]$ can be seen as those elements of $\RR$ in the form $f(0)(0)$
for a computable function $f\in\PrC$.
\subsection{Probabilistic Recursive Functions equals Functions computed by Probabilistic Turing Machines}
In this section we prove that probabilistic \emph{recursive} functions
are the same as probabilistic \emph{computable} functions, modulo an
appropriate bijection between strings and natural numbers which we
denote (as its inverse) with $\overline{(\cdot)}$.

In order to prove the equivalence result we first need to show that a
probabilistic recursive function can be computed by a PTM. This result
is not difficult and, analogously to the deterministic case, is proved
by exhibiting PTMs which simulate the basic probabilistic recursive
functions and by showing that $\PrC$ is closed by composition,
primitive recursion, and minimization.  \condinc {We omit the details,
  which can be found in~\cite{EV}.}  {This is done by the following
  Lemmata.

\begin{lemma}[Basic Functions are Computable]\label{Basic}
All Basic Probabilistic Functions are Computable.
\end{lemma}
 \begin{proof}
For every basic function from Definition~\ref{basprobfun}, we can construct a Probabilistic Turing Machine 
that computes it quite easily. More specifically, the proof is immediate for functions, $z, s, \Pi$, by 
observing that they are deterministic, thus the usual Turing machine for them (seen as a PTM). 
As for the function $\rand$ it can be simulated by a PTM $\ptmone$ which operates as follows:
\begin{varenumerate}
\item 
  $\ptmone$ deletes all the input written on the tape;
\item 
  $\ptmone$ writes $\overline{1}$ or $\overline{0}$ on the tape, both
  with probability $1/2$, and then halts.
\end{varenumerate}
This concludes the proof.
 \end{proof}
The composition of two computable probabilistic functions is itself computable:
\begin{lemma}[Generalized Composition and Computability] \label{Compo}
Given Turing-Computable $\f{f}{\NN^n}{\distrset{\NN}}$, 
and $\f{g_1}{\NN^k}{\distrset{\NN}},\ldots,\f{g_n}{\NN^k}{\distrset{\NN}}$  
the function $\f{f \odot (g_1,\ldots,g_n)}{\NN^k}{\distrset{\NN}}$ is
itself Turing-Computable
\end{lemma}
\begin{proof}
We give an informal proof. We define PTM, said $\ptmone$ working on $n+2$ tapes.
(We know that PTMs with $m>1$ tapes compute the same class of functions of PTMs with a single tape.) 
The first tape is the input tape, on the next $n+1$ tapes $\ptmone$ computes  
$g_1,\cdots, g_n$, while on the last tape, $\ptmone$ computes the function $f$ on the results of 
$g_1,\cdots, g_n$. 
The machine $\ptmone$ operates as follows:
\begin{varenumerate}
\item 
  it copies the input from the first to the next $n$ tapes; 
\item 
  in the $i+1$-th tape, the machine $\ptmone$ computes the respective function 
  $g_i$, where $1\le i \le n$; this can of course be done, because, by induction,
  the $g_i$ are computable;
\item  
  it copies the $n$ outputs in the $n$ tapes numbered $2,\ldots,n+1$ to the last tape;
\item 
  computes the function $f$ on the last tape and return the result $z$.
\end{varenumerate}
This concludes the proof.
\end{proof}

\begin{lemma}[Primitive Recursion and Turing-Computability] \label{Recur}
Given Turing-Computable $\f{g}{\NN^{k+2}} {\distrset{\NN}}$ and $\f{f}{\NN^k}{\distrset{\NN}}$
the function $\f{\primrec{f}{g}}{\NN^{k+1}}{\distrset{\NN}}$ is itself Turing-Computable.
\end{lemma}
\begin{proof}
We give an informal proof. We define PTM, said $\ptmone$ working on $5$ tapes.
The first tape is the input tape, on the next  tape $\ptmone$ computes the count down of our $k+1^{th}$ variable, on the third tape $\ptmone$ computes $g$, on the fourth tape $\ptmone$ computes the function $f$, and in the last saves the result.
The machine $\ptmone$ operates as follows:
\begin{varenumerate}
\item 
  it copies in the second tape the $k+1^{th}$ element of the input, and then it copies on the fourth tape the first $k$ elements of the input; 
\item it computes $f$ and saves the result on the last tape;
\item it verifies if the second tape is $0$. In this case  $\ptmone$ stops and the last tape contains the result, otherwise it copies the first $k$ elements of the input from the first tape in the third tape and then it copies the result present in the last tape on the third tape;  
\item $\ptmone$ decrements of the value on the second tape;
\item it computes $g$ on the third tape and save the result on the last tape;
\item it returns to the step 3.
\end{varenumerate}
This concludes the proof.

\end{proof}

\begin{lemma}[Minimization and Turing-Computability]\label{Minimiz}
Given Turing-Computable $\f{f}{\NN^{k+1}}{\distrset{\NN}}$, the function
$\f{\mu \; f} {\NN^k} {\distrset{\NN}}$ is itself Turing-Computable.
\end{lemma}

\begin{proof}
We give an intuitive proof.
We take a PTM, said $\ptmone$ with $4$ tapes.
The first tape is the input tape, on the next tape $\ptmone$ saves one element that we name $y$, on the third tape it computes the function $f$ and in the last tape it saves the result.
The machine $\ptmone$ operates as follows:
\begin{varenumerate}
\item it writes in the second tape $0$ and it copies on the third tape the input and the value $y$ (present in 
the second tape);
\item it computes on the third tape the function $f$ and saves the result on the last tape;
\item it verifies if the last tape contains the value $0$. In this case it saves on the last tape the element in the second tape and it stops, otherwise it increases  $y$;
\item it copies in the third tape the input and $y$; 
\item it returns to the step 3. 
\end{varenumerate}
This concludes the proof.

\end{proof}

Hence we can prove the following theorem, showing that Probabilistic Recursive Functions are computable by a Probabilistic Turing Machine. 
\begin{theorem}\label{theo:soundness} 
  $\PrR \subseteq \PrC$
\end{theorem}
\begin{proof}
Immediate from Lemmata \ref{Basic}, \ref{Compo}, \ref{Recur} and \ref{Minimiz}.
\end{proof}
}

The most difficult part of the equivalence proof consists in proving that 
each probabilistic computable function is actually \emph{recursive}.
Analogously to the classic case, a good strategy consists in representing
configurations as natural numbers, then encoding the transition functions 
of the machine at hand, call it $\ptmone$, as a (recursive) function on $\NN$.
In the classic case the proof proceeds by making essential use of the minimization operator 
by which one determines the \emph{number} of transition steps of $\ptmone$ necessary to reach 
a final configuration, if such number exists. This number can then be fed into another
function which simulates $\ptmone$ (on an input) a given number of steps, and which is primitive recursive.
In our case, this strategy does not work: the number of computation steps
can be infinite, even when the probability of converging is $1$. Given our definition of minimization 
which involves distributions, this is delicate, since we have to define a suitable function 
on the PTM computation tree to be minimized. 

In order to do adapt the classic proof, we need to formalize the
notion of a \emph{computation tree} which represents all computation
paths corresponding to a given input string $x$.  We define such a
tree as follows. Each node is labelled by a configuration of the
machine and each edge represents a computation step. The root is
labelled with $\INIT{M}{x}$ and each node labelled with $C$ has either
no child (if $C$ is final) or $2$ children (otherwise), labelled with
$\delta_0(C)$ and $\delta_1(C)$.
Please notice that the same configuration may be duplicated across a single level
of the tree as well as appear at different levels of the tree; nevertheless
we represent each such appearance by a separate node.

We can naturally associate a probability with each node, corresponding
to the probability that the node is reached in the computation: it is
$\frac{1}{2^n}$, where $n$ is the height of the node.  The probability
of a particular \emph{final} configuration is the sum of the
probabilities of all leaves labelled with that configuration.  We also
enumerate nodes in the tree, top-down and from left to right, by using
binary strings in the following way: the root has associated the
number $\varepsilon$. Then if $b$ is the binary string representing
the node $N$, the left child of $N$ has associated the string $b\cdot
0$ while the right child has the number $b\cdot 1$.  Note that from
this definition it follows that each binary number associated to a
node $N$ indicates a path in the tree from the root to $N$. The
computation tree for $x$ will be denoted as $\CT{M}{x}$

We give now a more explicit description of the constructions described
above. First we need to encode the rational numbers $\QQ$ into
$\NN$. Let $\f{\pair}{\NN \times \NN}{\NN}$ be any recursive bijection between
pairs of natural numbers and natural numbers such that $\pair$ and its
inverse are both computable. Let then $\enc$ be just $p_\pair$, i.e.
the function $\f{\enc}{\NN\times\NN}{\distrset{\NN}}$ defined as
follows
$$
\enc(a,b)(q)=\left\{ 
  \begin{array}{ll} 
    1 & \mbox{ if $q=\pair(a,b)$}\\
    0 & \mbox{ otherwise}
  \end{array} 
\right.
$$
The function $\enc$ allows to represent positive rational numbers as pairs of natural numbers in the
obvious way and is recursive.

It is now time to define a few notions on computation trees
\begin{definition}[Computation Trees and String Probabilities]\label{pt}
The function $\f{\PT{\ptmone}}{\NN\times\NN}{\QQ}$ is defined by stipulating
that $\PT{\ptmone}(x,y)$ is the probability of observing the string $\overline{y}$
in the tree $\CT{M}{x}$, namely $\frac{1}{2^{|\overline{y}|}}$.
\end{definition}
Of course, $\PT{\ptmone}$ is partial recursive, thus $p_{\PT{\ptmone}}$ is probabilistic recursive. 
Since the same configuration $C$ can label more than one node in a computation tree $\CT{M}{x}$, 
$\PT{\ptmone}$ does not indicate the probability of reaching $C$, even when $C$ is the label of
the node corresponding to the second argument. Such a probability can be obtained by
summing the probability of all nodes labelled with the same configuration at hand:
\newcommand{\CCM}[1]{\mathit{CC}_{#1}}
\begin{definition}[Configuration Probability]\label{pc}
Suppose given a PTM $\ptmone$. If $x\in\NN$ and $z\in\CONF{\ptmone}$, the
subset $\CCM{\ptmone}(x,z)$ of $\NN$ contains precisely the indices
of nodes of $\CT{\ptmone}{x}$ which are labelled by $z$. The function
$\f{\PC{\ptmone}}{\NN\times\NN}{{\QQ}}$ is defined as follows:
$$ 
\PC{\ptmone}(x,z) = \Sigma_{y \in \CCM{\ptmone}(x,z)} \PT{\ptmone}(x,y)
$$
\end{definition}
Contrary to $\PT{\ptmone}$, there is nothing guaranteeing that $\PC{\ptmone}$ is indeed computable.
In the following, however, what we will do is precisely showing that this is the case.

In Figure \ref{ctree} we show an example of computation tree $\CT{\ptmone}{x}$ for an hypothetical
PTM $\ptmone$ and an input $x$. The leaves, depicted as red nodes, represent the final configurations 
of the computation. So, for example, $\PC{\ptmone}(x,C)=1$, while
$\PC{\ptmone}(x,E)=\frac{3}{4}$. Indeed, notice that there are three nodes in the tree which
are labelled with $E$, namely those corresponding to the binary strings $00$, $01$, and $10$.
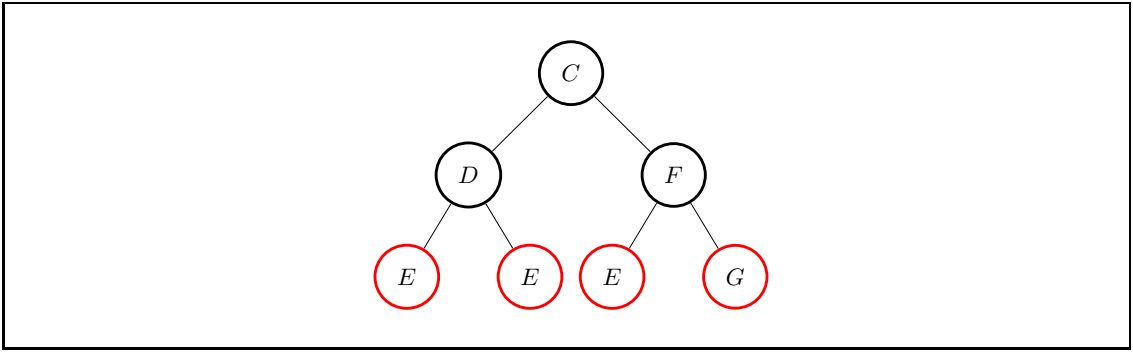
\begin{figure}[ht]
\begin{center}
\fbox{
\begin{minipage}{.97\textwidth}
\begin{center}
  \scalebox{0.9}{
    \begin{tikzpicture}[
	grow=down,
	level 1/.style={sibling distance=3cm, level distance=1.5cm},
	level 2/.style={sibling distance=1.8cm, level distance=1.5cm},
	level 3/.style={sibling distance=1.3cm, level distance=1.5cm},
	kant/.style={text width=3cm, text centered, sloped},
	every node/.style={text ragged, inner sep=2mm},
	punkt/.style={circle, shade, top color=white,
	  bottom color=white, draw=black, very
	  thick},
	punktn/.style={circle, shade, top color=white,
	  bottom color=white, draw=red, very
	thick},
      ]
      \node[punkt] {$C$}
      child {node[punkt] {$D$}
	child{node[punktn]{$E$}
	  edge from parent
	  node[kant,above,pos=.3]{}}
	child {node[punktn]{$E$}
	  edge from parent
	  node[kant,above,pos=.3]{}} 
	edge from parent
	node[kant,above,pos=.3]{}  }
      child{node[punkt] {$F$}
	child{node[punktn]{$E$}
	  edge from parent
	  node[kant,above,pos=.3]{}}     
	child {node[punktn]{$G$}
	  edge from parent
	  node[kant,above,pos=.3]{} 
	}	
        edge from parent
	node[kant,above,pos=.3]{} 
      };
    \end{tikzpicture}}
  \vspace{9pt}
  \end{center}
\end{minipage}}
\end{center}
\caption{An Example of a Computation Tree}
\label{ctree}
\end{figure}
As we already mentioned, our proof separates the classic part of the
computation performed by the underlying PTM, which essentially
computes the configurations reached by the machine in different paths,
from the probabilistic part, which instead computes the probability
values associated to each computation by using minimization. These two
tasks are realized by two suitable probabilistic recursive functions,
which are then composed to obtain the function computed by the
underlying PTM. We start with the probabilistic part, which is more
complicated.

We need to define a function, which returns the \emph{conditional}
probability of terminating at the node corresponding to the string
$\overline{y}$ in the tree $\CT{M}{x}$, given that all the nodes
$\overline{z}$ where $z<y$ are labelled with non-final
configurations. This is captured by the following definition:
\newcommand{\PTZ}[1]{\mathit{PT}^0_{#1}}
\newcommand{\PTU}[1]{\mathit{PT}^1_{#1}}
\begin{definition}\label{p1p0}
Given a PTM $\ptmone$, we define $\f{\PTZ{\ptmone}}{\NN \times \NN }{\QQ}$ and $\f{\PTU{\ptmone}}{\NN \times \NN }{\QQ}$   as follows:
\begin{align*}
\PTU{\ptmone}(x,y)&=  
\left\{ 
  \begin{array}{ll} 
    1 & \mbox{ if $y$ is not a leaf of $\CT{\ptmone}{x}$;}  \\
    1- \PTZ{\ptmone}(x,y) & \mbox{ otherwise;}
  \end{array} 
\right.
\\
\PTZ{\ptmone}(x,y)&=  
\left\{ \begin{array}{ll} 
    0 & \mbox{ if $y$ is not a leaf of $\CT{\ptmone}{x}$;} \\
    \frac{\PT{\ptmone}(x,y)}{\prod_{k< y}\PTU{\ptmone}(x,k)}  & \mbox{ otherwise;}
  \end{array} 
\right.
\end{align*}
\end{definition}
Note that, according to previous definition, $\PTU{\ptmone}(x,y)$ is
the probability of \emph{not} terminating the computation in the node
$y$, while $\PTZ{\ptmone}(x,y)$ represents the probability of terminating the computation in the node $y$, 
both \emph{knowing} that the computation has not terminated in any node $k$ preceding $y$. 
\begin{proposition}
The functions $\f{\PTZ{\ptmone}}{\NN \times \NN }{\QQ}$ and $\f{\PTU{\ptmone}}{\NN \times \NN }{\QQ}$ are partial recursive.
\end{proposition}
\begin{proof}
Please observe that $\PT{\ptmone}$ is partial recursive and that the definitions above are mutually recursive,
but the underlying order is well-founded. Both functions are thus intuitively computable, thus partial recursive 
by the Church-Turing thesis.\qed
\end{proof}
The reason why the two functions above are useful is because they
associate the distribution
$\{0^{\PTU{\ptmone}(x,y)},1^{\PTZ{\ptmone}(x,y)}\}$ to each pair of
natural numbers $(x,y)$.  In Figure \ref{xtree}, we give the
quantities we have just defined for the tree from Figure~\ref{ctree}.
Each internal node is associated with the same distribution
$\{0^{0},1^{1}\}$.  Only the leaves are associated with nontrivial
distributions.  As an example, the distribution associated to the node
$10$ is $\{0^{1/2},1^{1/2}\}$, because we have that
\begin{align*}
\PTZ{\ptmone}(x,\overline{10})&=\frac{\PT{\ptmone}(x,\overline{10})}{\prod_{k<\overline{10}}\PTU{\ptmone}(x,k)}\\
  &=\frac{1}{4\cdot\PTU{\ptmone}(x,\overline{01})\cdot\PTU{\ptmone}(x,\overline{00})\cdot\PTU{\ptmone}(x,\overline{1})\cdot\PTU{\ptmone}(x,\overline{0})\cdot\PTU{\ptmone}(x,\overline{\varepsilon})}\\
  &=\frac{1}{4\cdot\PTU{\ptmone}(x,\overline{01})\cdot\PTU{\ptmone}(x,\overline{00})}.
\end{align*} 
As it can be easily verified, $\PTU{\ptmone}(x,\overline{00})=\frac{3}{4}$, while $\PTU{\ptmone}(x,\overline{01})=\frac{2}{3}$.
Thus, $\PTZ{\ptmone}(x,\overline{10})=\frac{1}{2}$.
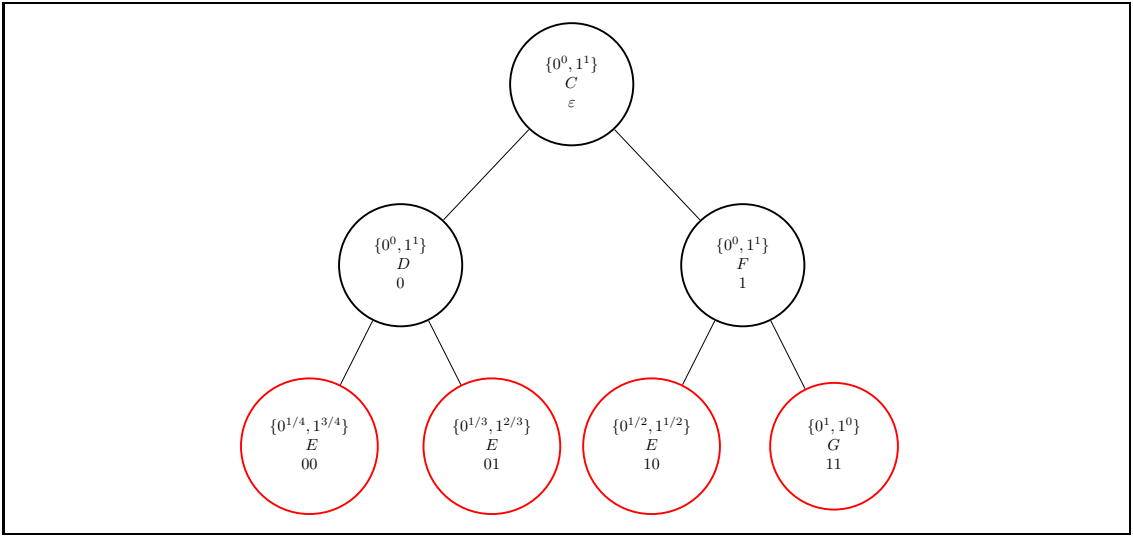
\begin{figure}[ht]
\begin{center}
\fbox{
\begin{minipage}{.97\textwidth}
  \begin{center}
    \vspace{4pt}
  \scalebox{0.6}{
    \begin{tikzpicture}[
      grow=down,
      level 1/.style={sibling distance=7.5cm, level distance=4cm},
      level 2/.style={sibling distance=4cm, level distance=4cm},
      level 3/.style={sibling distance=1.3cm, level distance=1.5cm},
      kant/.style={text width=3cm, text centered, sloped},
      every node/.style={text ragged, inner sep=2mm},
      punkt/.style={circle, minimum size=2.7cm, shade, top color=white,
	bottom color=white, draw=black, very
	thick},
      punktn/.style={circle,minimum size=2.8cm, shade, top color=white,
	bottom color=white, draw=red, very
	thick},
      ]
      \node[punkt] 
      {
        $\begin{array}{c}
          \{0^{0},1^{1}\}\\ 
          C\\
          \varepsilon
        \end{array}$}
      child {node[punkt] {
          $\begin{array}{c}
            \{0^{0},1^{1}\}\\ \ 
              D\\
              0
            \end{array}$}
          child{node[punktn]{$
              \begin{array}{c}
                \{0^{1/4},1^{3/4}\}\\ \ 
                E\\
                00
              \end{array}$}
            edge from parent
            node[kant,above,pos=.3]{}}
          child {node[punktn]{$
              \begin{array}{c}
                \{0^{1/3},1^{2/3}\}\\
                E\\ 
                01
              \end{array}$}
            edge from parent
            node[kant,above,pos=.3]{}} 
          edge from parent
          node[kant,above,pos=.3]{}  }
        child{node[punkt] {$\begin{array}{c}
              \{0^{0},1^{1}\}\\
              F\\
              1
            \end{array}$}
          child{node[punktn]{$\begin{array}{c}
                \{0^{1/2},1^{1/2}\}\\
                E\\
                10
              \end{array}$}
            edge from parent
            node[kant,above,pos=.3]{}}     
          child {node[punktn]{$\begin{array}{c}
                \{0^{1},1^{0}\}\\
                G\\
                11
              \end{array}$}
            edge from parent
            node[kant,above,pos=.3]{} 
	  }	
          edge from parent
          node[kant,above,pos=.3]{} 
        };
      \end{tikzpicture}}
  \vspace{4pt}
  \end{center}
  \end{minipage}}
\end{center}
\caption{The Conditional Probabilities for the Computation Tree from Figure \ref{ctree}}
\label{xtree}
\end{figure}

\newcommand{\PTC}[1]{\mathit{PTC}_{#1}}
We now need to go further, and prove that the probabilistic function returning,
on input $(x,y)$, the distribution $\{0^{\PTU{\ptmone}(x,y)},1^{\PTZ{\ptmone}(x,y)}\}$ is
recursive. This is captured by the following definition:
\begin{definition}\label{condprob}
Given a PTM $\ptmone$, the function $\f{\PTC{\ptmone}}{\NN \times \NN }{\distrset{\NN}}$ is defined as follows
$$
\PTC{\ptmone}(x,y)(z)=  
\left\{ 
  \begin{array}{ll} 
    \PTZ{\ptmone}(x,y) & \mbox{ if $z=0$;} \\
    \PTU{\ptmone}(x,y) & \mbox{ if $z=1$;}\\
    0 & \mbox{ otherwise-}
  \end{array} 
\right.
$$
\end{definition}
The function $\PTC{\ptmone}$ is really the core of our encoding. On the one hand, we will show that it is
indeed recursive. On the other, minimizing it is going to provide us exactly with the function we need
to reach our final goal, namely proving that the probabilistic function computed by $\ptmone$ is itself recursive. 
But how should we proceed if we want to prove $\PTC{\ptmone}$ to be recursive? The idea is to compose
$p_{\PTU{\ptmone}}$ with a function that turns its input into the probability of returning $1$.
This is precisely what the following function does:
\newcommand{\ItoP}{\mathit{I2P}}
\begin{definition} \label{j}
The function $\f{\ItoP}{\QQ}{\distrset{\NN}}$  is defined as follows
$$
\ItoP(x)(y)=\left\{ 
  \begin{array}{ll} 
    x & \mbox{ if $(0\leq x\leq 1)\wedge(y=1)$}\\
    1-x& \mbox{ if $(0\leq x\leq 1)\wedge(y=0)$}\\
    0 & \mbox{ otherwise}
  \end{array} 
\right.
$$
\end{definition}
Please observe how the input to $\ItoP$ is the set of rational numbers, as usual
encoded by pairs of natural numbers. Previous definitions allow us to treat 
(rational numbers representing) probabilities in  our algebra of functions. 
Indeed:
\begin{proposition}
The probabilistic function $\ItoP$ is recursive.
\end{proposition}
\begin{proof}
We first observe that $\f{h}{\NN}{\distrset{\NN}}$ defined as
$$
h(x)(y)=1/2^{y+1}
$$
is a probabilistic recursive function, because $h=\mu\;(\rand\odot \Pi^2_1)$.
Next we observe that every $q\in\QQ\cap[0,1]$ can be represented in
binary notation as:
$$
q=\sum_{i\in\NN} \frac{c^q_i}{1/2^{i+1}}
$$
where $c^q_i\in\{0,1\}$ (i.e., $c^q_i$ is the $i$-th element of the
binary representation of $q$). Moreover, a function computing such a
$c^q_i$ from $q$ and $i$ is partial recursive. Hence we can define
$\f{b}{\NN\times\NN}{\distrset{\NN}}$ as follows
$$
b(q,i)(y) =\left\{ 
  \begin{array}{ll} 
    1 & \mbox{if $y=c^q_i$}\\
    0 & \mbox{otherwise}
  \end{array} \right.
$$ 
and conclude that $b$ is indeed a probabilistic recursive function
(because $\PrR$ includes all the partial recursive functions, seen as
probabilistic functions). Observe that:
$$
b(q,i)(y) =\left\{ 
  \begin{array}{ll} 
    c^q_i & \mbox{if $y=1$}\\
    1-c^q_i & \mbox{if $y=0$}
  \end{array} \right.
$$ 
From the definition of composition, it follows that
\begin{align*}
(b \odot (id,h))(q)(y)&=\sum_{x_1,x_2} b(x_1,x_2)(y) \cdot id(q)(x_1)\cdot h(q)(x_2)\\
  &=\sum_{x_2} b(q,x_2)(y)\cdot h(q)(x_2)=\sum_{x_2} b(q,x_2)(y)\cdot \frac{1}{2^{x_2+1}}\\
  &=\left\{
  \begin{array}{ll}
    \sum_{x_2} \frac{c^q_{x_2}}{2^{x_2+1}} & \mbox{if $y=1$}\\
    \sum_{x_2} \frac{1-c^q_{x_2}}{2^{x_2+1}} & \mbox{if $y=0$}    
  \end{array}
  \right.
  =
  \left\{
  \begin{array}{ll}
    q & \mbox{if $y=1$}\\
    1-q & \mbox{if $y=0$}    
  \end{array}
  \right. .
\end{align*}
This shows that
$$
\ItoP= b \odot (id,h),
$$
and hence that $\ItoP$ is probabilistic recursive.\qed
\end{proof}
The following is an easy corollary of what we have obtained so far:
\begin{proposition}
The probabilistic function $\PTC{\ptmone}$ is recursive.
\end{proposition}
\begin{proof}
Just observe that $\PTC{\ptmone} = \ItoP\odot p_{\PTU{\ptmone}}$.\qed
\end{proof}
The probabilistic recursive function obtained as the minimization of
$\PTC{\ptmone}$ allows to compute a probabilistic function that, given
$x$, returns $y$ with probability $\PT{\ptmone}(x,y)$ \emph{if} $y$ is a
leaf (and otherwise the probability is just $0$).
\begin{definition} \label{p}
The function $\f{\CF{M}}{\NN }{\distrset{\NN}}$ is defined as follows
$$\CF{M}(x)(y)=  
\left\{ 
  \begin{array}{ll} 
    \PT{\ptmone}(x,y) & \mbox{if $y$ corresponds to a leaf}\\
    0 & otherwise.
  \end{array} 
\right.
$$
\end{definition}
\begin{proposition}
The probabilistic function $\CF{\ptmone}$ is recursive.
\end{proposition}
\begin{proof}
The probabilistic function $\CF{\ptmone}$ is just the function
obtained by minimizing $\PTC{\ptmone}$, which we already know to
be recursive. Indeed, if $z$ corresponds to a leaf, then:
\begin{align*}
(\mu\PTC{\ptmone})(x)(z)&=\PTC{\ptmone}(x,z)(0)\cdot\prod_{y<z}\sum_{k>0}\PTC{\ptmone}(x,y)(k)\\
   &=\PTC{\ptmone}(x,z)(0)\cdot\prod_{y<z}\PTC{\ptmone}(x,y)(1)\\
   &=\PTZ{\ptmone}(x,z)\cdot\prod_{y<z}\PTU{\ptmone}(x,y)\\
   &=\frac{\PT{\ptmone}(x,z)}{\prod_{y<z}\PTU{\ptmone}(x,y)}\cdot\prod_{y<z}\PTU{\ptmone}(x,y)=\PT{\ptmone}(x,z).
\end{align*}
If, however, $z$ does not correspond to a leaf, then:
\begin{align*}
(\mu\PTC{\ptmone})(x)(z)&=\PTC{\ptmone}(x,z)(0)\cdot\prod_{y<z}\sum_{k>0}\PTC{\ptmone}(x,y)(k)\\
  &=\PTZ{\ptmone}(x,z)(0)\cdot\prod_{y<z}\sum_{k>0}\PTC{\ptmone}(x,y)(k)=0.
\end{align*}
This concludes the proof.\qed
\end{proof}
\newcommand{\SP}[1]{\mathit{SP}_{#1}}
We are almost ready to wrap up our result, but before proceeding further, we need to define 
the function $\f{\SP{\ptmone}}{\NN\times\NN}{\NN}$ that, given in input a pair $(x,y)$ returns the 
(encoding) of the string found in the configuration labeling the node $y$ in $\CT{\ptmone}{x}$.
We can now prove the desired result:
\begin{theorem}\label{theo:completeness}
 $\PrC \subseteq \PrR$.
\end{theorem}
\begin{proof}
It suffices to note that, given any PTM $\ptmone$, the
function computed by $\ptmone$ is nothing more than
$$
p_{\SP{\ptmone}}\odot (\id,\CF{\ptmone}).
$$
Indeed, one can easily realize that a way to simulate $\ptmone$ indeed consists in generating
from $x$, all strings corresponding to the leaves of $\CT{\ptmone}{x}$, each with an appropriate
probability. This is indeed what $\CF{\ptmone}$ does. What remains to be done is simulating $p_{\SP{\ptmone}}$ along
paths leading to final configurations, which is what $\SP{\ptmone}$ does.\qed
\end{proof}
We are finally ready to prove the main result of this Section:
\begin{corollary} \label{corollarycoin}
$\PrR = \PrC$
\end{corollary}
\begin{proof}
Immediate from 
Theorem~\ref{theo:completeness}, observing that $\PrR \subseteq \PrC$ (this implication is easy to prove).\qed
\end{proof}
The way we prove Corollary \ref{corollarycoin} implies that we cannot deduce Kleene's Normal
Form Theorem from it: minimization has been used many times, some of them ``deep inside''
the construction. A way to recover Kleene's Theorem consists in replacing minimization with
a more powerful operator, essentially corresponding to computing the fixpoint of a given function
\condinc{(see~\cite{EV} for more details)}
{
}.

\section{Characterizing Probabilistic Complexity by Tiering}
In this section we provide a characterization of the probabilistic
functions which can be computed in polynomial time by an algebra
of functions acting on word algebras. More precisely, we define a type
system inspired by Leivant's notion of tiering~\cite{Leivant}, which
permits to rule out functions having a too-high complexity, thus
allowing to isolate the class of \emph{predicative probabilistic
functions}.
\condinc{ 
Finally, we give a hint at how the equivalence between polytime
probabilistic functions and predicative probabilistic functions can be
proved (more details are in~\cite{EV}).  }{ Our main result in this
section is that the class $\PrPC$ of probabilistic functions which
can be computed by a PTM in \emph{polynomial} time equals to
the class of predicative probabilistic functions.  }
\condinc
{

The constructions from Section~\ref{sect:prt} can be easily
generalized to a function algebra on strings in a given alphabet
$\alpone$, which themselves can be seen as a \emph{word algebra}
$\al{W}$.  Base functions include a function computing the empty
string, called $\varepsilon$, and concatenation with any character
$\charone\in\alpone$, called $c_\charone$. Projections remain of
course available, while the only truly random functions concatenate a
symbol $\charone\in\alpone$ to the input, with probability
$\frac{1}{2}$ or leave it unchanged, with probability
$\frac{1}{2}$. Such a function is denoted as $r_\charone$. Composition
and primitive recursion are available, although the latter takes the
form of recursion \emph{on notation}. We do not need minimization: the
distribution a polytime computable probabilistic function returns (on
any input) is always finite, and primitive recursion is anyway
powerful enough for our purposes.

The following construction is redundant in presence of primitive
recursion, but becomes essential when predicatively restricting it:
\begin{definition}[Case Distinction]
If $g_{\varepsilon}:\al{W}^k\rightarrow\distrset{\al{W}}$ and for
every $\charone\in\alpone$,
$g_{\charone}:\al{W}^{k+1}\rightarrow\distrset{\al{W}}$, the function
$h:\al{W}^{k+1}\rightarrow\distrset{\al{W}}$ such that
$h(\varepsilon,\V{y})=g_{\varepsilon}(\V{y})$ and $h(\charone \cdot
w,\V{y})= g_{\charone}(w, \V{y})$ is said to be defined by \emph{case
distinction} from $g_\varepsilon$ and
$\{g_{\charone}\}_{\charone\in\alpone}$ and is denoted as
$\casef{g_\varepsilon,\{g_{\charone}\}_{\charone\in\alpone}}$.
\end{definition}
The idea behind tiering consists in working with denumerable many
copies of the underlying algebra $\al{W}$, each indexed by a natural
number $n\in\NN$ and denoted by $\al{W}_n$. Judgments take the form
$f\triangleright\al{W}_{n_1}\times\ldots\times\al{W}_{n_k}\rightarrow\al{W}_m$,
where $f:\al{W}^k\rightarrow\al{W}$. In the following, with slight
abuse of notation, $\V{W}$ stands for any expression in the form
$\al{W}_{i_1} \times \cdots \times \al{W}_{i_j}$.

Typing rules are in Figure~\ref{fig:tieringrules}.
\begin{figure}
\begin{center}
\fbox{
\footnotesize
\begin{minipage} {.97\textwidth}
  $$
  \infer[]
  {\varepsilon\triangleright \al{W}_k}{}
  \quad
  \infer[]
  {r_\charone\triangleright \al{W}_k \rightarrow \al{W}_k}{}
  \quad
  \infer[]
  {c_\charone  \triangleright \al{W}_k \rightarrow \al{W}_k}{}
  \quad 
  \infer[]
  {\Pi_n^k  \triangleright  \al{W}_{n_1}\times \cdots \times  \al{W}_{n_n} \rightarrow \al{W}_{n_k}}{}
  $$
  \vspace{1pt}
  $$ 
  \infer[]
  {
    f\odot (g_1, \ldots, g_l) \triangleright  \al{W}_{s_1} \times \cdots \times \al{W}_{s_r} \rightarrow \al{W}_l\;
  }
  {
    \{g_i \triangleright   \al{W}_{s_1} \times \cdots \times \al{W}_{s_r} \rightarrow \al{W}_{m_i}\}_{1\leq i\leq p}
    &
    f \triangleright \al{W}_{m_1} \times \cdots \times \al{W}_{m_p} \rightarrow \al{W}_l
  }
  $$
  \vspace{1pt}
  $$ 
  \infer[]
  {
    case(g_{\varepsilon},\{g_\charone\}_{\charone\in\alpone})  \triangleright \al{W}_k \times \V{W}  \rightarrow \al{W}_l\;
  }
  {
    \begin{array}{c}
      g_{\varepsilon} \triangleright  \V{W} \rightarrow \al{W}_l\\
      \{g_\charone \triangleright   \al{W}_k \times \V{W}  \rightarrow \al{W}_l\}_{\charone\in\alpone}
    \end{array}
  } 
  \qquad
  \infer[]
  {
    rec(g_{\epsilon}, \{g_\charone\}_{\charone\in\alpone})  \triangleright \al{W}_m \times \V{W}  \rightarrow \al{W}_k\;
  }
  {
    \begin{array}{c}
    g_{\varepsilon} \triangleright \V{W} \rightarrow \al{W}_k\qquad m>k\\
    \{g_\charone \triangleright   \al{W}_k \times \al{W}_m \times \V{W}  \rightarrow \al{W}_k\}_{\charone\in\alpone}
   \end{array}
  } 
  $$
  \vspace{1pt}
\end{minipage}}
\end{center}
\caption{Tiering as a Typing System}\label{fig:tieringrules}
\end{figure}
The idea here is that, when generating functions by primitive
recursion, one goes from a level (tier) $m$ for the domain to
a \emph{strictly} lower level $k$ for the result. This predicative
constraint ensures that recursion does not cause any exponential
blowup, simply because the way one can \emph{nest} primitive recursive
definitions one inside the other is severely restricted. Please notice
that case distinction, although being typed in a similar way,
does \emph{not} require the same constraints.

Those probabilistic functions $f:\al{W}^k\rightarrow\distrset{\al{W}}$
such that $f$ can be given a type through the rules in
Figure~\ref{fig:tieringrules} are said to be \emph{predicatively
recursive}. The class of all predicatively recursive functions is
$\PrT$. Actually, the class coincides with the one of probabilistic
functions which can be computed by PTMs in polynomial time:
\begin{theorem}\label{theo:socotiering}
$\PrT=\PrPC$.
\end{theorem}
We don't have enough space to give the details of the proof of
Theorem~\ref{theo:socotiering}. It however proceeds essentially by
showing the following four lemmas, from which the thesis can be easily
inferred:
\begin{varitemize}
\item
  On the one hand, one can prove, by a careful encoding, that a form
  of \emph{simultaneous primitive recursion} is available in
  predicative recursion.
\item
  On the other hand, PTMs can be shown equivalent, in terms of
  expressivity, to probabilistic \emph{register} machines; going
  through register machines has the advantage of facilitating the last
  two steps.
\item
  Thirdly, any function definable by predicative recurrence can be
  proved computable by a polytime probabilistic register machine.
\item
  Lastly, one can give an embedding of any polytime probabilistic
  register machine into a predicatively recursive function, making use
  of simultaneous recurrence.
\end{varitemize}
Characterizing complexity classes of \emph{probabilistic} functions
allows to deal implicitly with concepts like that of
a \emph{polynomial time samplable}
distribution~\cite{BogdanovTrevisan06,Goldreich00}, which is a family
$\{\mathcal{D}_n\}_{n\in\NN}$ of distributions on strings such that a
polytime randomized algorithm produces $\mathcal{D}_n$ when fed with
the string $1^n$. By Theorem~\ref{theo:socotiering}, each of them is
computed by a function in $\PrT$ and, conversely, any predicatively
recursive probabilistic function computes one such family.
}
{

The constructions from Section~\ref{sect:prt} can be easily
generalized to a function algebra on strings in a given alphabet
$\alpone$, which themselves can be seen a \emph{word algebra}
$\al{W}$.  In the following we provide the details of such a
generalization.
\begin{definition}[String Distribution] \label{distr}
A \emph{pseudodistribution} on $\al{W}$ is a function
$D:\fun{\al{W}}{\RRI}$ such that $\sum_{w \in \al{W}}D(W)= 1$. The set
$\distrset{\al{W}}$ is defined as the set of all distribution on
$\al{W}$.
\end{definition}
The functions in our algebra have domain $\al{W}^k$ and codomain
$\distrset{\al{W}}$. The idea, as usual, is that $f(x)(y)=r$ means
that $y$ is the output obtained for the input $x$ with probability
$r$.  Base functions include a function computing the empty string,
denoted $\varepsilon$, and concatenation with any character
$\charone\in\alpone$, denoted by $c_\charone$. Formally we define
these functions as follows:
\begin{align*}
\varepsilon(v)(w)&=  
\left\{ 
  \begin{array}{rl} 
  1 & \mbox{if $w=\varepsilon$;}\\
  0 & \mbox{otherwise.}
  \end{array}  
\right.\\
c_{\charone} (v)(w)&=  
\left\{ \begin{array}{rl} 
1 & \mbox{if $w=\charone \cdot v$;}\\
0 & \mbox{otherwise.}
   \end{array} 
\right.
\end{align*}
Note that, for every $v \in \al{W}$, the length of the word obtained
after the application of one of the constructors $c_a$ is $|v|+1$
with probability $1$. Projections remain available in the usual
form. Indeed the function $\f{\Pi_m^n}{\al{W}^n}{\distrset{\al{W}}}$
is defined as follows:
$$
\Pi_n^k(\V{v})(w) =  
\left\{ \begin{array}{rl} 
1 & \mbox{if $w=v_m$;}\\
0 & \mbox{otherwise}.
\end{array}  
\right.
$$
The only truly random functions in our algebra are 
probabilistic functions in the form $\f{r_{\charone}}{\al{W}}{\distrset{\al{W}}}$,
which concatenates $\alpone$ to the input string (with probability $\frac{1}{2}$,
or leave it unchanged (with probability $\frac{1}{2}$). Formally,
$$
r_\charone(v)(w) =  
\left\{ \begin{array}{rl} 
1/2 & \mbox{if $w=\charone\cdot v$;}\\
1/2 & \mbox{if $w=v$;}\\
0 & \mbox{otherwise.}  
\end{array}         
\right.
$$
Next we recall the concept of composition and recurrence introduced in
Definition \ref{GComp} and Definition \ref{prec} and we instantiate
them to the case of our algebra. Generalized Composition of functions
$\f{f}{\al{W}^n}{\distrset{\al{W}}}$,
$\f{g_1}{\al{W}^k}{\distrset{\al{W}}},\ldots,\f{g_n}{\al{W}^k}{\distrset{\al{W}}}$
as the function 
$\f{f \odot (g_1,\ldots,g_n)}{\al{W}^k}{\distrset{\al{W}}}$ defined as:
$$
\alcomn{f}{(g_1,\ldots,g_n)}{\V{v}}{w}=\sum_{x_1,\ldots,x_n\in\al{W}}\left(f(x_1,\ldots,x_n)(w)\cdot\prod_{1\leq i\leq n} g_i(\V{v})(x_i)\right).
$$
Recurrence over $\al{W}$ takes the form:
\begin{align*}
f({\varepsilon}, \V{v})&=g_{\varepsilon}(\V{v})\\
f(\charone \cdot w,\V{v})&= g_{\charone}(f(w,\V{v}), w, \V{v})\qquad\forall\charone\in\alpone
\end{align*}
where $\f{f}{\al{W} \times \al{W}^m}{\distrset{\al{W}}}$, and
$\f{\charone}{\al{W} \times \al{W}\times \al{W}^m}{\distrset{\al{W}}}$
for every $\charone\in\alpone$. We use
$f=\mathit{rec}(g_{\epsilon},\{g_{\charone}\}_{\charone\in\alpone})$
as a shorthand for previous definition of recurrence. The following
construction is redundant in presence of primitive recursion, but
becomes essential when predicatively restricting it.
\begin{definition}[Case Distinction]\label{Case_Distinction}
If $g_{\varepsilon}:\al{W}^k\rightarrow\distrset{\al{W}}$ and for
every $\charone\in\alpone$,
$g_{\charone}:\al{W}^{k+1}\rightarrow\distrset{\al{W}}$, the function
$h:\al{W}^{k+1}\rightarrow\distrset{\al{W}}$ such that
$h(\varepsilon,\V{y})=g_{\varepsilon}(\V{y})$ and $h(\charone \cdot
w,\V{y})= g_{\charone}(w, \V{y})$ is said to be defined by \emph{case
distinction} from $g_\varepsilon$ and
$\{g_{\charone}\}_{\charone\in\alpone}$ and is denoted as
$\casef{g_\varepsilon,\{g_{\charone}\}_{\charone\in\alpone}}$.
\end{definition}
In the following we will need also the following by definition of
simultaneous recursion:
\begin{definition}\label{def:simrec}
We say that the functions $\V{f}=(f^1,\ldots,f^n)$ are defined by
\emph{simultaneous primitive recursion} over a word algebra $\al{W}$ from
the function $g^j_\charone$ (where $j\in\{1,\ldots,n\}$ and $\charone\in\alpone$
if the following holds for every $j$ and for every $\charone$:
$$
f^j(\charone\cdot v,\V{w})=g_{\charone}^j(f^1(v,\V{w}),\ldots,f^n(v,\V{w}),v,\V{w})
$$
A function $f^i$ as defined above will be indicated with
$\mathit{simrec}^i(\{g^j_\varepsilon\}_{j},\{g^j_\charone\}_{j,\charone})$.
\end{definition}
\begin{example}
This definition allows to define, for instance, two functions $f^0$
and $f^1$ over a word algebra with $\alpone=\{\charone,\chartwo\}$, as follows:
\begin{align*}
f^j (\epsilon, \V{v})&=g^j_\epsilon(\V{v})\qquad\forall j\in\{0,1\}\\
f^j(\charone\cdot w,\V{v})&=g_{\charone}^j(f^0(w,\V{v}), f^1(w,\V{v}), w, \V{v})\qquad\forall j\in\{0,1\}\\
f^j(\chartwo\cdot w,\V{v})&=g_{\chartwo}^j(f^0(w,\V{v}), f^1(w,\V{v}), w, \V{v})\qquad\forall j\in\{0,1\}
\end{align*}
\end{example}
\subsection{Tiering as a Typing System}\label{st}
Now we define our type system which will then be used to introduce the
definition of the class of \emph{predicative probabilistic functions}
and therefore to obtain our complexity result. The type system is
inspired by the tiering approach of Leivant \cite{Leivant}.  The idea
behind tiering consists in working with denumerable many copies of the
underlying algebra $\al{W}$, each indexed by a natural number
$n\in\NN$ and denoted $\al{W}_n$. Type judgments take the form
$f\triangleright\al{W}_{n_1}\times\ldots\times\al{W}_{n_k}\rightarrow\al{W}_m$,
where $f:\al{W}^k\rightarrow\al{W}$. In the following, with slight
abuse of notation, $\V{W}$ stands for any expression in the form
$\al{W}_{i_1} \times \cdots \times \al{W}_{i_j}$.  Typing rules are
given in Figure~\ref{fig:tieringrules}.
\begin{figure}
\begin{center}
\fbox{
\footnotesize
\begin{minipage} {.97\textwidth}
  $$
  \infer[]
  {\epsilon \triangleright \al{W}_k \rightarrow \al{W}_k}{}
  \quad
  \infer[]
  {c_\charone  \triangleright \al{W}_k \rightarrow \al{W}_k}{}
  \quad 
  \infer[]
  {r_\charone  \triangleright \al{W}_k \rightarrow \al{W}_k}{}
  \quad
  \infer[]
  {\Pi_m^n  \triangleright  \al{W}_{n_1}\times \cdots \times  \al{W}_{n_n} \rightarrow \al{W}_{n_m}}{}
  $$
  $$ 
  \infer[]
  {
    f\odot (g_1, \ldots, g_l) \triangleright  \al{W}_{s_1} \times \cdots \times \al{W}_{s_r} \rightarrow \al{W}_l\;
  }
  {
    \{g_i \triangleright   \al{W}_{s_1} \times \cdots \times \al{W}_{s_r} \rightarrow \al{W}_{m_i}\}_{1\leq i\leq l}
    &
    f \triangleright \al{W}_{m_1} \times \cdots \times \al{W}_{m_p} \rightarrow \al{W}_l
  }
  $$
  $$ 
  \infer[]
  {
    case(g_{\varepsilon},\{g_\charone\}_{\charone\in\alpone})  \triangleright \al{W}_k \times \V{W}  \rightarrow \al{W}_l\;
  }
  {
    \begin{array}{c}
      g_{\varepsilon} \triangleright  \V{W} \rightarrow \al{W}_l\\
      \{g_\charone \triangleright   \al{W}_k \times \V{W}  \rightarrow \al{W}_l\}_{\charone\in\alpone}
    \end{array}
  } 
  \qquad
  \infer[]
  {
    rec(g_{\epsilon}, \{g_\charone\}_{\charone\in\alpone})  \triangleright \al{W}_m \times \V{W}  \rightarrow \al{W}_k\;
  }
  {
    \begin{array}{c}
    g_{\varepsilon} \triangleright \V{W} \rightarrow \al{W}_k\qquad m>k\\
    \{g_\charone \triangleright   \al{W}_k \times \al{W}_m \times \V{W}  \rightarrow \al{W}_k\}_{\charone\in\alpone}
   \end{array}
  } 
  $$
\end{minipage}}
\end{center}
\caption{Tiering as a Typing System}\label{fig:tieringrules}
\end{figure}
The formulation of ramified recurrence over a probabilistic word
algebra $\al{W}$ derives from the definition of recurrence, suitably
restricted using types. The idea here is that, when generating
functions by primitive recursion, one passes from a level (tier) $m$
for the domain to a \emph{strictly} lower level $k$ for the
result. This predicative constraint ensures that recursion does not
causes complexity explosion.

Those probabilistic functions $f:\al{W}^k\rightarrow\distrset{\al{W}}$
such that $f$ can be given a type through the rules in
Figure~\ref{fig:tieringrules} are said to be \emph{predicatively
recursive}. More precisely, the class $\PrT$ of all predicatively
recursive functions is defined as follows.
\begin{definition}\label{classPrT}
The class $\PrT$ of \emph{predicatively probabilistic recursive
functions} is the smallest class of functions that contains the basic
functions and is closed under the operation of General Composition
(Definition \ref{GComp}), Primitive Recursion (Definition \ref{prec}),
Case Distinction (Definition \ref{Case_Distinction}) and such that each
function can be given a type through the rules in
Figure~\ref{fig:tieringrules}.
\end{definition}
Next we give the definition of the class of simultaneous recursive
functions $\SrR$.
\begin{definition}\label{classSrR}
The class $\SrR$ of \emph{simultaneous recursive functions} is the
smallest class of functions that contains the basic functions and is
closed under the operation of General Composition
(Definition \ref{GComp} ), Simultaneous Recursion
(Definition \ref{def:simrec}), Case Distinction
(Definition \ref{Case_Distinction} and such that each function can be
given a type through the rules in Figure~\ref{fig:tieringrules}, plus
the rule below:
 $$ 
   \infer []
  { \mathit{simrec}^i(\{g^j_\varepsilon\}_{j},\{g^j_\charone\}_{j,\charone})\triangleright \al{W}_m \times \V{W} \rightarrow \al{W}_k\;
  }
  { \begin{array}{c} \{g^j_{\epsilon} \triangleright \V{W} \rightarrow \al{W}_k\}_{j}\qquad
  m>k\\ \{g^j_\charone \triangleright \al{W}_k^n \times \al{W}_m \times \V{W} \rightarrow \al{W}_k\}_{j,\charone} \end{array}
  }
  $$
\end{definition}

\subsection{Simultaneous Primitive Recursion and Predicatie Recursion}
We can encode Simultaneous Primitive Recursion in Predicative Recursion. 

In fact, according top previous definition, if we have the two
functions $f^0, f^1$ over a word algebra with $\Sigma =\{c_0,c_1\}$,
defined by simultaneous recursion as follows:
$$
f^0 (\epsilon, \V{x})=g^0_\epsilon(\V{x})
$$
$$
f^0(c_0(w),\V{x})= g_{0}^0(f(w,\V{x})^0, f(w,\V{x})^1, w, \V{x})
$$
$$
f^0(c_1(w),\V{x})= g_{1}^0(f(w,\V{x})^0, f(w,\V{x})^1, w, \V{x})
$$
$$
f^1 (\epsilon, \V{x})=g^1_\epsilon(\V{x})
$$
$$
f^1(c_0(w),\V{x})= g_{0}^1(f(w,\V{x})^0, f(w,\V{x})^1, w, \V{x})
$$
$$
f^1(c_1(w),\V{x})= g_{1}^1(f(w,\V{x})^0, f(w,\V{x})^1, w, \V{x})
$$
we can see that we can define a function $\f{\widetilde{f}}{\al{W}^{k+1}}{\al{W}^{2}}$ as follows:
$$
\widetilde{f} (\epsilon, \V{x})=[g^1_\epsilon(\V{x}), g^0_\epsilon(\V{x})]
$$
$$
\widetilde{f}(c_1(w),\V{x})=[g_{1}^1(f(w,\V{x})^0, f(w,\V{x})^1, w, \V{x}),g_{1}^0(f(w,\V{x})^0, f(w,\V{x})^1, w, \V{x})]
$$
$$
\widetilde{f}(c_0(w),\V{x})=[ g_{0}^1(f(w,\V{x})^0, f(w,\V{x})^1, w, \V{x}),  g_{0}^0(f(w,\V{x})^0, f(w,\V{x})^1, w, \V{x})]
$$
Now the codomain of $\widetilde{f}$ can be coded in a single value of $\al{W}$, this function is said $couple_m$. Inversely we can define two functions, $first_m$ and $second_m$ that given a value in $\al{W}$ $first_m$ returns the first value of the couple, and the $second_m$ returns the second values.

Now we proof as this encoding uses at most $m$ time the recurrence.

Firstly we define the length of a word as follows.

\begin{definition}
We define $|w|$ as the height of the parse-tree of the word $w$.
\end{definition}

\begin{lemma}\label{ENCO}
Let $\al{W}$ be word algebra, $m>0$, and using at most $1$ level of recurrence. 
There are functions ($\f{couple_m}{\al{W}\times \al{W}}{\al{W}}$, $\f{first_m}{\al{W}\times \al{W}}{\al{W}}$, $\f{second_m}{\al{W}\times \al{W}}{\al{W}}$) such that:

$$couple_m(u,v)=t$$
$$first_m(t) =u$$ 
$$second_m(t)=v$$
whenever $2|u|+2|v|+2\leq |t|^m$
\end{lemma}

\subsection{Register Machines vs. Turing Machines}

Register machines are a class of  abstract computational model which, 
when properly defined, are Turing powerful. 
Here we extend the classical  definition of Register Machines to the probabilistic case.


\begin{definition}[Probabilistic Register Machine]\label{PRM}
A Probabilistic Register Machine $(PRM)$ consists of a 
finite set of registers $\Pi=\{\pi_1, \cdots,\pi_r\}=\V{\pi}$  which can store elements belonging to $\al{W}$ 
and of a finite sequence of indexed instructions (called program) which can have the following format:
\begin{description}
\item[$\epsilon(\pi_s)(\pi_l)$]
\item[$c_i(\pi_s)(\pi_l)$] 
\item[$p_i(\pi_s)(\pi_l)$] 
\item[$jump(\pi_s)(\V{m})$] 
\item[$jump_{\rand} (m)$]
\end{description}
where $\pi_s, \pi_l$ are registers, $m$ is the number of the instruction and $c_\charone$ is any constructor of the algebra $\al{W}$ and $\V{m}$ is a vector of $|\Sigma|$ elements of natural number.
\end{definition}

The semantic of previous instructions can be described as follows. We assume that the index of the current instruction is $n$.
\begin{description}
\item[$\epsilon(\pi_s)(\pi_l)$] is the $\epsilon$ instruction, which stores in the register  $\pi_l$ the term resulting from the application of the constructor $\epsilon$ to the register $(\pi_s)$ and then transfer the control to the next instruction $(n+1)$.
\item[$c_{\charone}(\pi_s)(\pi_l)$] is the constructor instruction, which stores in the register  $\pi_l$ the term resulting from the application of the constructor $c_{\charone}$ to the register $(\pi_s)$ and then transfer the control to the next instruction $(n+1)$.
\item[$p_{\charone}(\pi_s)(\pi_l)$] is the predecessor instruction which stores in the register $\pi_l$ the element resulting from the application of the predecessor $p_{\charone}$ to the register $\pi_s$, and then transfer the control to the next instruction $(n+1)$.
\item[$jump(\pi_s)(\V{m})$] is the jump instruction which: jumps to the instruction $m_{\overline{\charone}}$ and stores the result of $p_{\charone}(\pi_s)$ in $\pi_s$  if apply $c_{\charone}$ 
is the first constructor in the value contained in $\pi_s$;  transfer the control to the next instruction $n+1$  if $\pi_s$ contains $\epsilon$.
\item[$jump_{\rand} (m) $] is the jump randomize instruction which jumps to the instruction $m$ with probability $1/2$, or transfer the control to the next  instruction $n+1$ with probability $1/2$.
\end{description}

Below we formalize more precisely the semantics of a PRM in terms of configurations which can be modified by the instructions.
Hence we introduce the following,

\begin{definition}[Configuration of a PRM] \label{ConfigPRM}
Let $\prmone$  be a PRM be as in \ref{PRM}. We define a PRM configuration as a tuple 
$<v_1,\ldots,v_r,n>$ where:
\begin{itemize}
\item the $v_i$s are the values of the registers;
\item $n$ is a natural number indicating the current instruction.
\end{itemize}
We define the set of all configurations with $\CONFR{\prmone}$.
When $n=1$ we have an initial configuration 
for $r$ strings $\V{\strone}$ is indicated with $\INITR{\prmone}{\V{\strone}}$. When $n =max+1$, 
where $max$ is the number of instructions in the program, we have a final configuration said $\FSR{\prmone}{\V{\strone}}$.
\end{definition}

Next we show how previous instructions allow to change a configuration.
\begin{description}

\item[$\epsilon(\pi_s)(\pi_l)$] If we apply the instruction $\epsilon(\pi_s)(\pi_l)$ to the configuration  $<p_1,\ldots, p_r, n>$, we obtain the configuration  $<p_1,\ldots, p_{l-1},p_s ,p_{l+1},\ldots, p_r, n+1>$.

\item[$c_{\charone}(\pi_s)(\pi_l)$] If we apply the instruction $c_{\charone}(\pi_s)(\pi_l)$ to the configuration  $<p_1,\ldots, p_r, n>$, we obtain the configuration  $<p_1,\ldots, p_{l-1}, \charone \cdot p_s ,p_{l+1},\ldots, p_r, n+1>$.

\item[$p_{\charone}(\pi_s)(\pi_l)$]  If we apply the instruction $p_{\charone}(\pi_s)(\pi_l)$ to the configuration  $<p_1,\ldots, a \cdot p_s,\ldots, p_r, n>$, we obtain the configuration  $<p_1, \ldots, {\charone} \cdot p_s, \ldots,  p_{l-1},p_s, p_{l+1},\ldots, p_r, n+1>$.

\item[$jump(\pi_s)(\V{m})$] If we apply $jump(\pi_s),(\V{m})$ to the configuration  $<p_1,\ldots, a \cdot p_s,\ldots, p_r, n>$, we obtain 
the configuration $<p_1,\ldots, p_s,\ldots, p_r, m_{\overline{a}}>$;  If we apply $jump(\pi_s),(\V{m})$ to the configuration  $<p_1,\ldots, p_{s-1},\epsilon,p_{s+1},\ldots, p_r, n>$, we obtain 
the configuration $<p_1,\ldots, p_{s-1},\epsilon,p_{s+1},\ldots, p_r, n+1>$.

\item[$jump_{\rand} (m) $] If we apply $jump_{rand},s$ to the configuration  $<p_1,\ldots,p_r, n>$, we obtain 
the configuration $<p_1,\ldots,p_r, m>$  with probability $1/2$ and the configuration $<p_1,\ldots,p_r, n+1>$  with probability $1/2$.

\end{description}


Note that, according to the previous definition, one 
of the instructions of the machine is $jump_{\rand}$, which gives to the machine the probabilistic 
behavior. So in order to simulate the  behavior of  a PTM we can assume that 
we apply $jump_{\rand}$ after any instruction of another type.

Intuitively, the function computed by a PRM $\prmone$ associates to each input $\strone^r$
 a (pseudo)-distribution which indicates the probability of reaching a configuration in $\FSR{\prmone}{\V{\strtwo}}$
from $\INITR{\prmone}{\V{\strone}}$. It is worth noticing that, differently from the deterministic case, since in a PRM 
the same configuration can be obtained by different computations, the probability 
of reaching a given final configuration is the \emph{sum} of the probabilities of reaching the configuration 
along all computation paths, of which there can be (even infinitely) many. 
So also in this case it is convenient to define the function computed by a PRM through a fixpoint construction, as follows.
First we observe that the meaning of a PRM $\prmone$ program can be defined by using two functions $\delta_0$ and $\delta_1$. In fact 
 as previously mentioned we assume that in each PRM program we use  $jump_{\rand}$ after any instruction of another type. Hence 
 we can consider two functions $\f{\delta_0}{\CONFR{\prmone}}{\CONFR{\prmone}}$ and $\f{\delta_1}{\CONFR{\prmone}}{\CONFR{\prmone}}$ which, given a configuration in input, both produce in output the (unique) configuration resulting 
 from the application of an instruction, when this is different from $jump_{\rand}$. 
 When, on the other hand, $jump_{\rand}$ is used, $\delta_0$ and $\delta_1$ produce respectively the two configurations 
 (with probability $1/2$) resulting from the two branches of the instruction, as previously defined.

We can define a (complete) partial order on the $\confeval$ elements analogously to what we have done for the $PTM$. 
Hence we can now define a functional $FR_R$ on $\confeval$ which will be used to define the function computed by $R$ 
via a fixpoint construction. Intuitively, the application of the functional $FR_R$ describes \emph{one} computation step.
Formally we have the following:
\begin{definition}\label{def:functionalPRM}
Given a PRM $\prmone$, we define a functional $\f{FR_R}{\confeval}{\confeval} $ as:
$$
FR_R(f)(C)= \left\{ \begin{array}{ll} 
      \{\strone^1\}& \mbox{ if } C\in\FSR{M}{\V{\strone}};\\
      \frac{1}{2}f(\delta_0(C))+\frac{1}{2}f(\delta_1(C)) & \mbox{ otherwise}.
    \end{array} \right.
$$
\end{definition}

Using similar arguments to those of Proposition \ref{cont} and Theorem \ref{mon} we can show that 
there exists the least fixpoint of the functional defined above.
Such a least fixpoint, once composed with a function returning $\INITR{R}{\V{\strone}}$ from
$\V{\strone}$, is the \emph{function computed by the register machine $R$} and it is denoted by 
$\FUNCR{R}:\alpone^*\rightarrow\distrset{\alpone^*}$. The set of those functions which
can be computed by any PRMs is denoted by $\PrT$. Moreover, we denote the class of the function computed by a polynomial time register machine as $\PPrT$



Next Lemma shows the relations between PTMs and PRMs.

\begin{lemma}\label{PTMvsPRM}
PTMs are linear time reducible to a PRMs, and PRMs over $\al{W}$ are poly-time reducible to PTMs. 
\end{lemma}
\begin{proof}
A single tape PTM $\ptmone$ can be simulated by a PRM $\prmone$ that has tree registers. A configuration $<w,a,v,s>$ of $\ptmone$ can be coded by the configuration $\prmone$ $<[w^r,a,v],s>$ where $s^r$ denotes the reverse of the string $s$. 
Each move of $\ptmone$ is simulated by at most $2$ moves of $\prmone$. 
In order to simulate the probabilistic part given by the functions $\delta_0$ and $\delta_1$ we use 
and instruction $jump_{rand}$ by assuming, for example, that if $jump_{rand}$ allows to jump to $m$ we simulate $delta_0$ (that is at the index $m$ we have the PRM code simulating $\delta_0$) otherwise we simulate  $\delta_1$.   
Conversely, a PRM $\prmone$ over $\al{W}$ with $m$ registers is simulated by an $m-tapes$ PTM $\ptmone$. Some move of $\prmone$ may require copying the contents of one register to another for which $\ptmone$ may need as many steps to complete as the maximum of the current lengths of the corresponding tapes. Thus $\prmone$ runs in time $O(n^k)$, then $\ptmone$ runs in time $O(n^{2k})$.
\end{proof}

\subsection{Poly-time Soundeness}
In this section we prove that any function definable by predicative recurrence is computable by a polynomial time probabilistic register machine said PPRM. 

From the lemma \ref{PTMvsPRM} we  can derive  that $\PPrT= \PrPC$.
Hence, in order to prove the equivalence result we first show that a predicative recurrence can be computed by a 
PPRM. This result is not difficult and is proved by exhibiting a PPRM which simulate 
the basic predicative recurrence functions and by showing that $\PPrT$ is closed by composition, primitive recursion, and case distinction, 
that is, we can construct a PPRM which simulates these operations (on the machine representing  predicative recurrence functions).

We start with the following Lemmata.

\begin{lemma}[Basic Functions and PPRM Computability]\label{PRMBasic}
All Basic Probabilistic Functions are computable by a PPRM.
\end{lemma}
\begin{proof}
We need to show that for every basic function defined in Definition \ref{basprobfun} we can construct a PPRM that computes such a function. 
The proof is immediate for functions, $c_{\charone}$, by observing that it is included in the set of PPRM operations. 
The function $\epsilon$ is simulated by using the instruction $\epsilon$ (on an empty register).
The function $\Pi$ is simulated by the instructions $\epsilon(\pi_s)(\pi_l)$.
Finally the function $\rand$ can be simulated by the instructions $jump_{\rand}$ and $c_{\charone}$.  
\end{proof}

\begin{lemma}[Generalized Composition and PPRM Computability] \label{PRMCompo}
Given PPRM-computable $\f{f}{\NN^n}{\distrset{\NN}}$, 
and $\f{g_1}{\NN^k}{\distrset{\NN}},\ldots,\f{g_n}{\NN^k}{\distrset{\NN}}$  
the function $\f{f \odot (g_1,\ldots,g_n)}{\NN^k}{\distrset{\NN}}$ is
itself PPRM computable
\end{lemma}
\begin{proof}
We give an intuitive proof.
We take a PPRM, said $\prmone_s$ with $s$ registers.
The first $k$ registers have saved the input, the next $k \cdot n$ registers $\prmone_s$ computes the $g_1,\cdots, g_n$ functions. For each $g_i$ the machine saves the result on the registers $(k \cdot n) + i$, with $1 \leq i \leq n$. These registers became the input registers for computing $f$. Finally in the last register is saved the result.
 $\prmone_s$  operates as follows:
\begin{varenumerate}
\item $\prmone_s$ copies all $k$ registers on the $k \cdot n$ registers. The computational cost of this operation is $n\cdot k$, because it is implemented by the instruction $\epsilon(\pi_{l})(\pi_{s})$;
\item $\prmone_s$ computes the respective functions $g_i$ with $1 \leq i \leq n$ and saves the results in the registers $(k \cdot n) + i$. These functions are by hypothesis polynomial time computable;
\item $\prmone_s$ computes the function $f$ that by hypothesis polynomial time computable.
\end{varenumerate}
Finally $\prmone_s$ computes the function $f \odot (g_1,\ldots,g_n)$ in time 
\begin{align*}
z &= k \cdot n + \sum_{i=1}^n( max(|pi_i|)^{s_i}+v_i )+ max(pi_{(k \cdot n) + i})^t+w \\
    &\leq  (k \cdot n) + (n +1) \cdot (max (|pi_i|, |pi_{(k \cdot n) + i}|)^{max(s_{i},t)}+max(v_i,w))\\
    &< (n +1)\cdot (y^m+q +k)
\end {align*}
where $pi$ denotes the element saved on the register, $max (|pi_i|, |pi_{(k \cdot n) + i}|)=y, max(s_i,t)=m$ and $max(v_i,w) = q$.
\end{proof}

\begin{lemma}[Case Distinction and PPRM-Computability] \label{PRMCase}
Given PPRM-Computable $\f{g_{\{{\charone}\}_{\charone\in\alpone}}}{\NN^{k+2}} {\distrset{\NN}}$, and $\f{g_\varepsilon}{\NN^k}{\distrset{\NN}}$,
the function $\casef{g_\varepsilon,\{g_{\charone}\}_{\charone\in\alpone}}$ is itself PPRM-Computable.
\end{lemma}
\begin{proof}
The function Case Distinction is implemented by a PPRM, said $\prmone_{case}$, that computes it as follows. 
The first $k+1$ registers contain the input and on the last register we have the result. 
$\prmone_{case}$ operates as follows:
\begin{varenumerate}
\item $\prmone_{case}$ applies the operation $jump(\pi_{k+1})(\V{m})$;  if $pi_{k+1} = \epsilon$ the machine goes to the instruction $i+1$ where there is saved the first instruction in order to compute $g_\varepsilon$, otherwise the machine jump at the instruction $m_{\charone}$ corresponding at the first (in $pi_{k+1}$) constructor function $c_{\charone}$ and it saves the result on $pi_{k+1}$ register. The instruction $m_{\charone}$ is the first one which allows  to compute $g_{\charone}$.
\item $\prmone_{case}$ computes the function $g_\varepsilon$ or $g_{\charone}$ and it saves the result on the last register.
\end{varenumerate}
Finally $\prmone_{case}$ computes the function Case Distinction  in time $z$, and $z$ is a polynomial time because:
\begin{align*}
z &= c + \left\{ \begin{array}{rl} 
 |s_{\varepsilon}|^{k_\varepsilon}+r_{\varepsilon} & $if $ pi_{k+1}=\epsilon  \\
 |s_{\charone}|^{k_\charone}+r_{\charone} & $ if $ pi_{k+1}=p_{\charone}
  \end{array} \right.
\end {align*}
where $c$ is the time constant used from the machine for computing the function $jump$.

\end{proof}

\begin{lemma}[Primitive Recursion and PPRM-Computability] \label{PRMRecur}
Given PPRM-Computable $\f{g}{\NN^{k+2}} {\distrset{\NN}}$, and $\f{f}{\NN^k}{\distrset{\NN}}$,
the function $\f{\primrec{f}{g}}{\NN^{k+1}}{\distrset{\NN}}$ is itself PPRM-Computable.
\end{lemma}
\begin{proof}
We give an intuitive proof.
We take a PPRM, said $\prmone_{rec}$.
The first $k+1$ registers contain the input, in the next $k$ registers we save the input, in the $2k+2^{th}$ register we have the result of intermediate computations, and in the last register we have the result. We assume that the $k+1^{th}$ input is saved inverted on the $k+1^{th}$ register.
$\prmone_{rec}$ operates as follows:
\begin{varenumerate}
\item  $\prmone_{rec}$ computes $g_\varepsilon$ and saves the results in the $2k+2^{th}$;
\item $\prmone_{rec}$ applies an operation of $jump(\pi_{k+1}),(\V{m})$;
\item  if $pi_{k+1} = \epsilon$ the machine goes to the instruction $i+1$ where it saves the result on the last registers and then stop, otherwise the machine jump to the instruction $m_{\charone}$  corresponding at the function predecessor $p_{\charone}$ and it saves the result on $pi_{k+1}$ register. $m_{\charone}$  is the first instruction needed in order to compute $g_{\charone}$.
\item $\prmone_{rec}$ jump at the instruction $2$ if it is not stopped before.
\end{varenumerate}
Finally $\prmone_{rec}$ computes the function $\primrec{f}{g}$ in time $z$, and $z$ is a polynomial time because:
\begin{align*}
z &= c \cdot |p_{k+1}| + |p_{k+1}| \cdot \left\{ \begin{array}{rl} 
 |s_{\varepsilon}|^{k_\varepsilon}+r_{\varepsilon} & $if $ pi_{k+1}=\epsilon  \\
 |s_{\charone}|^{k_\charone}+r_{\charone} & $ if $ pi_{k+1}=p_{\charone}
  \end{array} \right.\\
  & \leq  max (|p_{k+1}|,|s_{\varepsilon}|,|s_{\charone}|) \cdot(c+max (r_{\varepsilon},r_{\charone})
) + max (|p_{k+1}|,|s_{\varepsilon}|,|s_{\charone}|)\cdot max (|p_{k+1}|,|s_{\varepsilon}|,|s_{\charone}|)^{max(k_\varepsilon, k_\charone)}\\
&=  x ^ t + x \cdot d
\end {align*}
where $c$ is the time constant used from the machine for computes the function $jump$, $x=max (|p_{k+1}|,|s_{\varepsilon}|,|s_{\charone}|)$, $t=max(k_\varepsilon, k_\charone)+1$ and $d=c+max (r_{\varepsilon},r_{\charone})$.

\end{proof}

\begin{lemma}[PPRM Computable Functions are Predicative Recurrence Functions]\label{PrPCPrT}
$\PrT \subseteq  \PPrT$
\end{lemma}
\begin{proof}
We have that is proved by Lemma \ref{PRMBasic}, \ref{PRMCompo}, \ref{PRMRecur} and \ref{PRMCase}.
\end{proof}

\begin{lemma}[PPTM Computable Functions are Predicative Recurrence Functions]\label{PrTPrPC}
$\PrT \subseteq  \PrPC$ 
\end{lemma}
\begin{proof}
We have that is proved by Lemma \ref{PrPCPrT} and \ref{PTMvsPRM}.
\end{proof}

\subsection{Poly-time Completeness}

In this section, one can give an embedding of any polynomial time probabilistic register machine into a predicatively recursive function, making use of simultaneous recurrence.

In order to do this firstly we give some lemma in order to construct the proof.

We start defining some predicatively recursive functions in order to simulate a single step of a PPRM.
 
\begin{lemma}\label{aaa}
Let $\prmone$ be a PPRM with $r$ registers over $\al{W}$, there are functions $\f{\phi_0}{\al{W}^k}{\distrset{\al{W}}}, \cdots, \f{\phi_r}{\al{W}^k}{\distrset{\al{W}}}$, with tier at least one, such that if $\prmone$ has a transition rule from the configuration $i$ then
$(\V{pi}, n) \rightarrow_R (\V{pi}, m)$  if and only if $\phi_i(\V{pi_i},\overline{n}) = pi_i^{'}$ for all $1\leq i \leq r$ and $\phi_0(\V{u},\overline{n})=\overline{n}$
\end{lemma}
\begin{proof}
We observe that at every step each register's value $pi_i$ can assume only one of these values $(c_\charone(pi_i), p_\charone(pi_i),\epsilon(pi_i))$. 
We note that the natural number of the instruction can be creased or decreased of a finite number.
Applying the instruction $jump_{rand}$ we obtain a change in our configuration in the number of the next instruction and the relative probability.
Now we show as the instructions of our program are mapped by one or more predicatively recursive functions.

\begin{itemize}
\item starting from the configuration  $<p_1,\ldots, p_r, n>$ and applying $\epsilon(\pi_s)(\pi_l)$ we obtain  $<p_1,\ldots, p_{l-1},p_s ,p_{l+1},\ldots, p_r, n+1>$. So it can be mapped by these functions $\phi_s = \Pi^{r+1}_s$, $\phi_l = \Pi^{r+1}_s$, and $\phi_0 =c_\charone\odot (\Pi^{r+1}_0)$;

\item starting from the configuration $<p_1,\ldots, p_r, n>$ and applying$c_{\charone}(\pi_s)(\pi_l)$, we obtain the configuration  $<p_1,\ldots, p_{l-1}, \charone \cdot p_s ,p_{l+1},\ldots, p_r, n+1>$, it can be mapped by these functions $\phi_s = \Pi^{r+1}_s$, $\phi_l = c_{\charone} \odot (\Pi^{r+1}_s )$, and $\phi_0 =c_\charone \odot (\Pi^{r+1}_0)$;

\item starting form the configuration $<p_1,\ldots, a \cdot p_s,\ldots, p_r, n>$ and applying $p_{\charone}(\pi_s)(\pi_l)$ we obtain $<p_1, \ldots, {\charone} \cdot p_s, \ldots,  p_{l-1},p_s, p_{l+1},\ldots, p_r, n+1>$, it can be mapped by these functions $\phi_s = \Pi^{r+1}_s$, $\phi_l = p_{\charone} \odot (\Pi^{r+1}_s )$, and $\phi_0 =c_\charone \odot (\Pi^{r+1}_0)$;

\item starting from the configuration $<p_1,\ldots, a \cdot p_s,\ldots, p_r, n>$ and applying $jump(\pi_s),(\V{m})$  $<p_1,\ldots, p_s,\ldots, p_r, m_{\overline{a}}>$ we obtain 
the configuration $<p_1,\ldots, p_{s-1},\epsilon,p_{s+1},\ldots, p_r, n+1>$
it can be mapped by these functions $\phi_s = p_{\charone} \odot (\Pi^{r+1}_s )$, for each $m_{\overline{\charone}}$ with a function $\phi_0$ that if  $m_{\overline{\charone}}-n \geq 0$ has exactly $m_{\overline{\charone}}-n$ concatenation of constructors otherwise it has exactly $m_{\overline{\charone}}-n$ predecessor functions, finally each of these functions are  composed with the function $\casef{g_\varepsilon,\{g_{\charone}\}_{\charone\in\alpone}}$;

\item starting from the configuration  $<p_1,\ldots,p_r, n>$ and applying $jump_{\rand}$ we obtain $<p_1,\ldots,p_r, m>$  with probability $1/2$ and the configuration $<p_1,\ldots,p_r, n+1>$  with probability $1/2$, it can be mapped by the function $\rand$, and $\casef{g_\varepsilon,\{g_{\charone}\}_{\charone\in\alpone}}$. We note that in this case, the output of our register machine defines a distribution. 

\end{itemize}
We note that all these functions are predicatively recursive functions, and every instructions are mapped with functions with at most flat recurrence, and so they can be mapped with tier $0$.

\end{proof}

\begin{lemma}\label{bbb}
If a function $f$ over a word algebra $\al{W}$ is a computable by a PPRM $\prmone$  in time $t \leq C \dot n^k$,  then it is definable by $k$ Simultaneous Recurrence Functions over $\distrset{\al{W}}$ with tier at least one, applied to functions over $\distrset{\al{W}}$
\end{lemma} 
\begin{proof}
We start our proof by considering the case $C = 1$ and then we extend the proof to the case $C>1$.
We assume that $\prmone$ has $r$ registers. Now we define the functions $\f{\sigma_{qj}}{\al{W}^{q+m+2}}{\distrset{\al{W}}}$, with $0\leq q\leq k$ and $j= 0, \ldots r$, with tier at most one. The functions $(\sigma_{q1}, \cdots, \sigma_{qr})$ represent the values of the registers after $|y_1|\dot \cdots \dot |y_q| +|x|$ execution steps of $\prmone$, starting from the initial configuration. The function $\sigma_{q0}$ represents the index of the instruction after 
$|y_1|\dot \cdots \dot |y_q| +|x|$ execution steps.

For each register and  index instruction we define a function with at most $k$ nested recursion as follows:
\begin{align*}
\sigma_{0j}(\epsilon, (\V{u},\overline{n}))&=(\V{u},\overline{n})\\
\sigma_{0j}(\charone \cdot w, (\V{u},\overline{n} )&=\phi_{j}(\sigma_{0j}(\V{u},\overline{n}))\\
\sigma_{q+1,j}(\epsilon, \V{y},x,(\V{u},\overline{n} ))&=\sigma_{0j}(w, x,(\V{u},\overline{n}))\\
\sigma_{q+1,j}(\charone \cdot w,\V{y},x,(\V{u},\overline{n} ))&=\sigma_{q,j}(\sigma_{0j}(\V{y},x,(\V{u},\overline{n} )),\sigma_{q+1} (w, \V{y},\epsilon,(\V{u},\overline{n}))\\
\end{align*}
Now it suffices to use simultaneous recurrence for composing the function define on every register and this prove the thesis.

\end{proof}

\begin{lemma}[Class of Simultaneous Recurrence Functions is the Class of Probabilistic Register Computable Functions]
\label{SRRPRPC}

$\PrPC \subseteq   \SrR$
\end{lemma}
\begin{proof}
We have that is proved by Lemma \ref{aaa}, Lemma \ref{bbb} 
\end{proof}

\begin{theorem}[Class of Predicative Probabilistic Functions is the Class of Probabilistic Register Computable Functions]
$\PrPC = \PrT$
\end{theorem}
\begin{proof}
We have that is proved by  \ref{PrPCPrT}, Lemma \ref{SRRPRPC} and Lemma \ref{ENCO}.
\end{proof}

}

\section{Conclusions}
In this paper we make a first step in the direction of characterizing
probabilistic computation in itself, from a recursion-theoretical
perspective, without reducing it to deterministic computation. The
significance of this study is genuinely foundational: working with
probabilistic functions allows us to better understand the nature of
probabilistic computation on the one hand, but also to study the
implicit complexity of a generalization of Leivant's predicative
recurrence, all in a unified framework.

More specifically, we give a characterization of computable
probabilistic functions by a natural generalization of Kleene's
partial recursive functions which includes, among initial functions,
one that returns the uniform distribution on $\{0,1\}$. We then prove
the equi-expressivity of the obtained algebra and the class of
functions computed by PTMs. In the the second part of the paper, we
investigate the relations existing between our recursion-theoretical framework
and sub-recursive classes, in the spirit of ICC. More precisely,
endowing predicative recurrence with a random base function is proved
to lead to a characterization of polynomial-time computable
probabilistic functions.

An interesting direction for future work could be the extension of our
recursion-theoretic framework to \emph{quantum} computation.  In this
case one should consider transformations on Hilbert spaces as the
basic elements of the computation domain. The main difficulty towards
obtaining a completeness result for the resulting algebra and proving
the equivalence with quantum Turing machines seems to be the
definition of suitable recursion and minimization operators
generalizing the ones described in this paper, given that qubits (the
quantum analogues of classical bits) cannot be copied nor erased.

\condinc{
\bibliographystyle{abbrv}
\bibliography{biblio}
}
{
\bibliographystyle{abbrv}
\bibliography{biblio}

\begin{thebibliography}{10}

\bibitem{BellantoniCook}
S.~Bellantoni and S.~Cook.
\newblock A new recursion-theoretic characterization of the polytime functions.
\newblock {\em Computational complexity}, 2(2):97--110, 1992.

\bibitem{BogdanovTrevisan06}
A.~Bogdanov and L.~Trevisan.
\newblock Average-case complexity.
\newblock {\em Foundations and Trends in Theoretical Computer Science}, 2(1),
  2006.

\bibitem{DalLagoParisenToldin}
U.~Dal~Lago and P.~P. Toldin.
\newblock A higher-order characterization of probabilistic polynomial time.
\newblock In {\em Foundational and Practical Aspects of Resource Analysis},
  pages 1--18. Springer, 2012.

\bibitem{EV}
U.~Dal~Lago and S.~Zuppiroli.
\newblock Probabilistic recursion theory and implicit computational complexity
  (long version).
\newblock \url{http://eternal.cs.unibo.it/prtic.pdf}, 2014.

\bibitem{DeLeeuw53}
K.~De~Leeuw, E.~F. Moore, C.~E. Shannon, and N.~Shapiro.
\newblock Computability by probabilistic machines.
\newblock {\em Automata studies}, 34:183--198, 1956.

\bibitem{gill77}
J.~Gill.
\newblock Computational complexity of probabilistic {T}uring machines.
\newblock {\em SIAM Journal on Computing}, 6(4):675--695, 1977.

\bibitem{girard98}
J.-Y. Girard.
\newblock Light linear logic.
\newblock {\em Inf. Comput.}, 143(2):175--204, 1998.

\bibitem{Goldreich00}
O.~Goldreich.
\newblock {\em Foundations of Cryptography: Basic Tools}.
\newblock Cambridge University Press, 2000.

\bibitem{GoldwasserMicali}
S.~Goldwasser and S.~Micali.
\newblock Probabilistic encryption.
\newblock {\em Journal of computer and system sciences}, 28(2):270--299, 1984.

\bibitem{KatzLindell07}
J.~Katz and Y.~Lindell.
\newblock {\em Introduction to Modern Cryptography (Chapman \& Hall/Crc
  Cryptography and Network Security Series)}.
\newblock Chapman \& Hall/CRC, 2007.

\bibitem{kleene36}
S.~C. Kleene.
\newblock General recursive functions of natural numbers.
\newblock {\em Mathematische Annalen}, 112(1):727--742, 1936.

\bibitem{Leivant}
D.~Leivant.
\newblock Ramified recurrence and computational complexity i: Word recurrence
  and poly-time.
\newblock In {\em Feasible Mathematics II}, pages 320--343. Springer, 1995.

\bibitem{Leivant93}
D.~Leivant and J.-Y. Marion.
\newblock Lambda calculus characterizations of poly-time.
\newblock {\em Fundam. Inform.}, 19(1/2):167--184, 1993.

\bibitem{Rabin63}
M.~O. Rabin.
\newblock Probabilistic automata.
\newblock {\em Information and control}, 6(3):230--245, 1963.

\bibitem{rabinscott1959}
M.~O. Rabin and D.~Scott.
\newblock Finite automata and their decision problems.
\newblock {\em IBM J. Res. Dev.}, 3(2):114--125, 1959.

\bibitem{santos69}
E.~S. Santos.
\newblock Probabilistic {T}uring machines and computability.
\newblock {\em Proceedings of the American Mathematical Society},
  22(3):704--710, 1969.

\bibitem{santos71}
E.~S. Santos.
\newblock Computability by probabilistic turing machines.
\newblock {\em Transactions of the American Mathematical Society},
  159:165--184, 1971.

\bibitem{soare1987}
R.~I. Soare.
\newblock {\em Recursively enumerable sets and degrees: a study of computable
  functions and computably generated sets}.
\newblock Perspectives in mathematical logic. Springer-Verlag, 1987.

\end{thebibliography}
}
\end{document}